\pdfoutput=1
\RequirePackage{amsthm}

\documentclass[pdflatex,sn-basic,Numbered]{sn-jnl}

\usepackage{graphicx}%
\usepackage{multirow}%
\usepackage{amsmath,amssymb,amsfonts}%
\usepackage{amsthm}%
\usepackage{mathrsfs}%
\usepackage[title]{appendix}%
\usepackage{xcolor}%
\usepackage{booktabs}%
\usepackage{mathtools}
\usepackage{nicefrac}
\usepackage{microtype}
\usepackage{cleveref}
\usepackage{orcidlink}
\usepackage{csquotes}
\usepackage{marvosym}

\usepackage[shortlabels]{enumitem}
\usepackage{placeins}
\usepackage{xspace}
\usepackage{amssymb}

\theoremstyle{thmstyleone}%
\newtheorem{theorem}{Theorem}%
\newtheorem{observation}[theorem]{Observation}
\newtheorem{corollary}[theorem]{Corollary}

\theoremstyle{thmstyletwo}%
\newtheorem{remark}{Remark}%
\newtheorem{prop}{P\!\!}

\theoremstyle{thmstylethree}%
\newtheorem{definition}{Definition}%

\colorlet{mygreen}{green!85!black}
\colorlet{myred}{red!70!white}

\usepackage{subcaption}
\Crefname{lemma}{Lemma}{Lemmas}
\Crefname{observation}{Observation}{Observations}
\Crefname{corollary}{Corollary}{Corollaries}
\Crefname{definition}{Definition}{Definitions}
\Crefname{remark}{Remark}{Remarks}
\Crefname{case}{Case}{Cases}
\Crefname{prop}{P}{Properties}
\crefformat{prop}{P#2#1#3}
\setlist[description]{font=\normalfont\bfseries}

\usepackage[abbreviations,shortcuts,nohypertypes={abbreviations,problem}]{glossaries-extra}
\newglossary*{problem}{Problems}
\newcommand{\newproblem}[3]{
    \newglossaryentry{pr:#1}{type=problem,
        name={\ensuremath{#2}},
        description={#3},
        sort={#1}
    }
}
\newcommand{\prob}[1]{\glsentryname{pr:#1}}
\newcommand{\probf}[1]{\glsentrydesc{pr:#1} (\glsentryname{pr:#1})}
\newcommand{\probl}[1]{\glsentrydesc{pr:#1}}

\newabbreviation{DAG}{DAG}{directed acyclic graph}
\newabbreviation{DSP}{DSP}{directed series-parallel graph}
\newabbreviation{EAS}{EAS}{edge-anchored subgraph}
\newabbreviation{MEAS}{MEAS}{maximal edge-anchored subgraph}
\newabbreviation{FPTAS}{FPTAS}{Fully Polynomial Time Approximation Scheme}
\newabbreviation{FPT}{FPT}{Fixed Parameter Tractable}
\newabbreviation{MED}{MED}{minimum equivalent digraph}
\newabbreviation{MST}{MST}{minimum spanning tree}
\newabbreviation{LSP}{LSP}{laminar series-parallel graph}

\newproblem{DHC}{\textsf{DHC}}{\textsf{Directed Hamiltonian Cycle}}
\newproblem{UHC}{\textsf{UHC}}{\textsf{Undirected Hamiltonian Cycle}}
\newproblem{MCPS}{\textsf{MCPS}}{\textsf{Minimum Capacity-Preserving Subgraph}}
\newproblem{MCPS1}{\textsf{MCPS}_\textsf{unit}}{\textsf{Minimum Capacity-Preserving Subgraph with Unit Edge Capacities}}
\newproblem{MCPS*}{\textsf{MCPS}^*}{\textsf{Minimum Capacity-Preserving Subgraph, Beyond MED}}
\newproblem{MCPS*1}{\textsf{MCPS}_\textsf{unit}^*}{\textsf{Minimum Capacity-Preserving Subgraph$^*$, Unit Edge Capacities}}
\newproblem{MGCPS}{\textsf{MGCPS}}{\textsf{Minimum Global Capacity-Preserving Subgraph}}
\newproblem{MSCS}{\textsf{MSCS}}{\textsf{Minimum Strongly Connected Subgraph}}
\newproblem{MED}{\textsf{MED}}{\textsf{Minimum Equivalent Digraph}}
\newproblem{SC}{\textsf{SC}}{\textsf{Set Cover}}
\newproblem{VC}{\textsf{VC}}{\textsf{Vertex Cover}}
\newproblem{MST}{\textsf{MST}}{\textsf{Minimum Spanning Tree}}
\newproblem{DSND}{\textsf{DSND}}{\textsf{Directed Survivable Network Design}}

\let\union\cup
\let\intersect\cap
\let\Union\bigcup

\renewcommand{\epsilon}{\varepsilon}

\newcommand{\maxflow}[1]{\ensuremath{c_{#1}}}
\newcommand{\ecap}{\ensuremath{\textit{cap}}\xspace}
\newcommand{\gcap}[1]{\ensuremath{\mathit{gc}(#1)}}

\newcommand{\distsym}{\ensuremath{\mu}}
\newcommand{\dist}[1]{\ensuremath{\mu_{#1}}}
\newcommand{\freq}{\ensuremath{f}}
\newcommand{\stretcha}{\alpha}
\newcommand{\ssep}{\mid}
\newcommand{\lin}{+}
\newcommand{\lout}{-}
\newcommand{\pathind}[2]{\ensuremath{#1\langle#2\rangle}}

\newcommand{\medval}{\ensuremath{\operatorname{med}(G)}}

\newcommand{\univ}{\ensuremath{U}}
\newcommand{\setcover}{\ensuremath{\mathcal{C}}}
\newcommand{\sets}{\ensuremath{\mathcal{S}}}
\newcommand{\reds}{\ensuremath{\mathcal{R}}}
\newcommand{\greens}{\ensuremath{\mathcal{G}}}
\newcommand\subdiv[1]{\ensuremath{\bar{#1}}}

\newcommand{\bigO}{\ensuremath{\mathcal{O}}}

\newcommand{\R}{\ensuremath{\mathbb{R}}}
\newcommand{\N}{\ensuremath{\mathbb{N}}}
\newcommand{\be}{\ensuremath{\coloneqq}}

\DeclarePairedDelimiter\ceil{\lceil}{\rceil}

\usepackage{algorithm}
\usepackage[noend]{algpseudocode}
\newcommand*\Let[2]{\State #1 $\gets$ #2}
\algnewcommand\Break{\State \textbf{break}}
\algnewcommand\Continue{\State \textbf{continue}}
\algdef{SE}[DOWHILE]{Do}{DoWhile}{\algorithmicdo}[1]{\algorithmicwhile\ #1}%
\algrenewcommand\Return{\State \algorithmicreturn{} }
\algrenewcommand\algorithmicforall{\textbf{for each}}
\algnewcommand\algorithmicinput{\textbf{Input:}}
\algnewcommand\algorithmicoutput{\textbf{Output:}}
\algnewcommand\Input{\item[\algorithmicinput]}%
\algnewcommand\Output{\item[\algorithmicoutput]}%
\algrenewcommand\algorithmicthen{}
\algrenewcommand\algorithmicdo{}

\usepackage{tabularx}
\newcommand{\problemhelper}[2]{%
    \medskip\par\noindent%
    \renewcommand{\arraystretch}{1.2}%
    \renewcommand{\tabcolsep}{2pt}%
    \begin{tabularx}{\linewidth}{|>{\bfseries}lX|}%
    \hline%
    \multicolumn{2}{|X|}{%
        \rule{0pt}{3ex}\probf{#1}%
    }\\[1.7pt]%
    #2
    \hline%
    \end{tabularx}%
    \medskip\par%
}
\newcommand{\problem}[4]{%
    \problemhelper{#1}{%
    Input: & \begin{minipage}[t]{\linewidth}#2\end{minipage}\\%
    Feasible solution: & \begin{minipage}[t]{\linewidth}#3\end{minipage}\\%
    Optimization: & \begin{minipage}[t]{\linewidth}Find minimum #4 over all feasible solutions.\end{minipage}\\%
    Decision variant: & \begin{minipage}[t]{\linewidth}Also given $\kappa\in\mathbb{N}$, is there a feasible solution with \mbox{#4\ $\leq\kappa$}?\end{minipage}\\[1.2pt]%
    }%
}
\newcommand{\problemshort}[2]{%
    \problemhelper{#1}{%
    \multicolumn{2}{|X|}{#2}\\[1.2pt]%
    }%
}
\newcommand{\problemdouble}[5]{%
    \problemhelper{#1}{%
    Input: & \begin{minipage}[t]{\linewidth}#2\end{minipage}\\%
    Feasible solution: & \begin{minipage}[t]{\linewidth}#3\end{minipage}\\%
    Optimization: & \begin{minipage}[t]{\linewidth}Find minimum #4 over all feasible solutions.\end{minipage}\\%
    Decision variant: & \begin{minipage}[t]{\linewidth}Also given $\kappa\in\mathbb{N}$, is there a feasible solution with \mbox{#4\ $\leq\kappa$}?\end{minipage}\\[11.7pt]%
    \hline%
    \multicolumn{2}{|l|}{\rule{0pt}{3ex}#5}\\[1.2pt]%
    }%
}

\usepackage{tikz}
\usetikzlibrary{arrows,backgrounds,shapes,calc,decorations,decorations.pathreplacing,decorations.pathmorphing,positioning}
\tikzset{main_edge/.style={line width=2.3pt}}
\tikzset{overlap/.style={draw=magenta}}
\tikzset{dot/.style={circle,fill=black,inner sep=1pt,minimum size=1pt}}
\tikzset{path/.style={
    line join=round,
    decorate, decoration={
        zigzag,
        segment length=6,
        amplitude=.7,
        post=lineto,
        pre length=3pt,
        post length=3pt
    }, very thick,
    shorten <= 3pt,
    shorten >= 3pt
}}
\tikzset{dsp/.style={
    draw=black,
    thick,
    double=black,
    ->,
    shorten >= 2pt,
    shorten <= 2pt,
    line cap=round
}}
\tikzset{cycle-shade/.style={
    draw=lightgray!50!white,
    line width=15pt,
    rounded corners=5pt,
    line cap=round
}}
\tikzset{edge/.style={
    rounded corners, ->, draw=black, very thick, shorten <= 2pt,shorten >= 2pt}}
\tikzset{label/.style={
    fill opacity=0, draw opacity=0, text opacity=1, inner sep=0, outer sep=3.5pt}}

\newcommand\rnew[1]{#1}
\newcommand\rdel[1]{\ignorespaces}

\begin{document}

\title[Directed Capacity-Preserving Subgraphs]{Directed Capacity-Preserving
Subgraphs: Hardness~and Exact~Polynomial~Algorithms%
\footnote{This work was supported by the German Research Foundation (DFG) grant
461207633 (CH~897/7-1).}}

\author{\fnm{Markus} \sur{Chimani}\orcidlink{0000-0002-4681-5550}}
\email{markus.chimani@uos.de}

\author{\fnm{Max} \sur{Ilsen}\orcidlink{0000-0002-4532-3829}}
\email{max.ilsen@uos.de}

\affil{\orgdiv{Theoretical Computer Science}, \orgname{Osnabrück University},
\orgaddress{\city{Osnabrück},~\country{Germany}}}

\abstract{We introduce and discuss the
\probf{MCPS} problem: given
a directed graph with edge capacities $\ecap$ and a retention ratio
$\stretcha \in (0,1)$, find the
smallest subgraph that, for each pair of vertices~$(u,v)$, preserves at
least a fraction $\stretcha$ of a maximum $u$-$v$-flow's value.
This problem originates from the practical setting of reducing
the power consumption in a computer network:
it models turning off as many links as possible, while retaining
the ability to transmit at least $\stretcha$ times the traffic compared to
the original network.

First we prove that \prob{MCPS} is NP-hard already on a restricted set of \acp{DAG} with
unit edge capacities.
Our reduction also shows that a closely related problem (which only
considers the arguably most complicated core of the problem in the objective
function) is NP-hard to approximate within a sublogarithmic factor already
on \acp{DAG}.
In terms of positive results, we present two algorithms that solve \prob{MCPS}
optimally on \acp{DSP}: a simple linear-time algorithm for the special case of
unit edge capacities and a cubic-time dynamic programming algorithm for the
general case of non-uniform edge capacities.
Further, we introduce the family of \acp{LSP}, a generalization of \acp{DSP}
that also includes cyclic and very dense graphs.
Their properties allow us to solve \prob{MCPS} on \acp{LSP} by employing our
\ac{DSP}-algorithms as subroutines.
In addition, we give a separate quadratic-time algorithm for \prob{MCPS} on
\acp{LSP} with unit edge capacities that also yields straightforward quadratic
time algorithms for several related problems such as \probl{MED} and \probl{DHC}
on \acp{LSP}.
}


\maketitle

\section{Introduction}

We present the \probf{MCPS} problem. Interestingly, despite it being very
natural, simple to formulate, and practically relevant,
there seems to have been virtually no explicit research regarding it.
We may motivate the problem by recent developments in Internet usage and routing research:
Not only does Internet traffic grow rapidly~\cite{DBLP:journals/pieee/EssiambreT12,website/Cisco20,DBLP:journals/networks/Wong21},
current Internet usage shows distinct traffic peaks in the evening (when people are, e.g., streaming videos) and lows at night and in the early
morning~\cite{DBLP:journals/ton/SchullerACHS18,DBLP:conf/imc/FeldmannGLPPDWW20}.
This has sparked research into the reduction of power consumption for backbone Internet providers~(Tier 1) by turning off
unused resources~\cite{DBLP:conf/icnp/ZhangYLZ10,DBLP:conf/icc/ChiaraviglioMN09}:
One natural way is to turn off as many connections between servers as possible, while still retaining the ability
to route the occurring traffic. Typically, one assumes that (a) the original routing network is suitably dimensioned
and structured for the traffic demands at peak times, and (b) the traffic demands in the low times are mostly similar to the peak demands but
\enquote{scaled down} by some ratio.

Graph-theoretically, we are given a \rnew{(}directed\rnew{)} graph
$G=(V,E)$, a capacity function~$\ecap \colon E \to \R^+$ on its edges, and
a retention ratio~$\stretcha \in (0,1)$.
\rnew{Throughout this paper,} all graphs are simple (i.e., contain no self-loops nor parallel edges)\rnew{, and---unless specified otherwise---directed; we may use the term \emph{digraph} to stress the directedness of a graph}.
For every pair of vertices~$(s,t) \in V^2$, let
$\maxflow{G}(s,t)$ denote the value of a maximum flow (or equivalently, a
minimum cut) from~$s$ to~$t$ in~$G$ according to the capacity function~$\ecap$.
Thus, in the following we unambiguously refer to $\maxflow{G}(s,t)$ as the
\emph{capacity} of the vertex pair~$(s,t)$ in~$G$,
which intuitively represents how much flow can be sent from~$s$ to~$t$ in~$G$.
The lowest capacity among all vertex pairs corresponds to the size~$\gcap{G}$
of the global minimum cut.
\rdel{One may ask for an edge-wise minimum subgraph $G'=(V,E')$, $E'\subseteq E$, such
that $\gcap{G'} \geq \stretcha\cdot\gcap{G}$. We call this problem \probf{MGCPS}:}%
\rnew{Finding an edge-wise minimum subgraph that preserves a fraction~$\stretcha$ of the global minimum cut is NP-hard:}
\rnew{\problem{MGCPS}
{(Di)graph
$G=(V,E)$, edge capacities~$\ecap \colon E \to \R^+$, $\stretcha \in (0,1)$.}
{Subgraph $G'=(V,E')$, $E'\subseteq E$, such that $\gcap{G'} \geq \stretcha\cdot\gcap{G}$.}
{$|E'|$}
Throughout this paper, the problem's name and its abbreviation, both in sans-serif typeface, always refer to the optimization question. Using normal typeface, we may say that a subgraph is a MGCPS to clarify that it is an optimal solution.}

\begin{observation}
    \prob{MGCPS} is NP-hard, both on directed and undirected
    graphs, already with unit edge capacities.
\end{observation}
\begin{proof}
    \rdel{Identifying a Hamiltonian cycle in a directed strongly connected (or undirected
    2-edge-connected) graph~$G$ is~NP-hard~\cite{DBLP:conf/coco/Karp72}.
    Consider an optimal~\prob{MGCPS} solution for~$G$ with unit edge capacities and
    $\stretcha = \nicefrac{1}{\gcap{G}}$ ($\nicefrac{2}{\gcap{G}}$):
    every vertex pair is precisely required to have a capacity
    of at least $\ceil{\stretcha \cdot \gcap{G}} = 1$
    ($\ceil{\stretcha \cdot \gcap{G}} = 2$).
    Hence, an $\stretcha$-\prob{MGCPS} of $G$ must also be strongly connected
    (2-edge-connected, respectively) and
    is a Hamiltonian cycle if and only if one exists in $G$.}
    \rnew{%
    We start with the directed case.
    Given a directed graph $G$, the \probl{DHC} problem asks whether $G$ contains a directed cycle traversing every vertex of $G$ exactly once.
     A directed graph is \emph{strongly connected} if there exists a directed path from every vertex to every other vertex.
    \probl{DHC} is NP-hard even in strongly connected graphs~\cite{DBLP:conf/coco/Karp72}.
    From such a digraph $G=(V,E)$, we can construct an~\prob{MGCPS} instance~$(G, \ecap, \stretcha)$ with unit
    edge capacities and $\stretcha = \nicefrac{1}{\gcap{G}}$. %
    In an optimal solution for this instance, every vertex pair is precisely
    required to have a capacity of at least $\ceil{\stretcha \cdot \gcap{G}} = 1$.
    Hence, an optimal solution must be spanning and strongly connected.
    However, the smallest spanning strongly connected subgraph of~$G$ is a Hamiltonian cycle if and only if one exists in~$G$. Optimally solving \prob{MGCPS} would thus solve \probl{DHC}.

    The undirected case is similar:
    Identifying an undirected Hamiltonian cycle in a given undirected 2-edge-connected
    graph~$G$ is~NP-hard as well~\cite{DBLP:conf/coco/Karp72}.
    Thus, we can construct an~\prob{MGCPS} instance~$(G, \ecap, \stretcha)$
    with unit edge capacities and $\stretcha = \nicefrac{2}{\gcap{G}}$:
    in an optimal solution, every vertex pair is precisely required to have a
    capacity of at least $\ceil{\stretcha \cdot \gcap{G}} = 2$.
    Hence, an optimal solution must be spanning and 2-edge-connected,
    and the smallest such subgraph is a Hamiltonian cycle if and only if one exists in $G$.
    }
\end{proof}

However, in our practical scenario, \prob{MGCPS} is not so interesting.
\rdel{Thus, w}\rnew{W}e rather consider the problem where the capacities
$\maxflow{G}(u,v)$ have to be retained for each vertex pair~$(u,v)$
individually:
\rdel{In the \probf{MCPS} problem,
we ask for a set of edges~$E' \subseteq
E$ with minimum size~$|E'|$ yielding the subgraph $G' = (V,E')$, such that
$\maxflow{G'}(s,t) \geq \stretcha \cdot \maxflow{G}(s,t)$ for all $(s,t) \in V^2$.}
\rnew{\problem{MCPS}
{Digraph
$G=(V,E)$, edge capacities~$\ecap \colon E \to \R^+$, $\stretcha \in (0,1)$.}
{Subgraph $G'=(V,E')$, $E'\subseteq E$, such that\\
\begin{minipage}{\linewidth}\centering $\maxflow{G'}(s,t)\geq\stretcha \cdot \maxflow{G}(s,t)$ for each $(s,t) \in V^2$.\end{minipage}
}{$|E'|$}}%
For an \prob{MCPS} instance~$(G, \ecap, \stretcha)$, we will call a
vertex pair $(s,t)$ (or edge $st$) \emph{covered} by an
edge set $E'$ if the graph~$G' = (V,E')$ satisfies $\maxflow{G'}(s,t) \geq
\stretcha \cdot \maxflow{G}(s,t)$.

In the following, we often discuss the special
setting where the capacity function~$\ecap$ assigns 1 to every edge---in this
setting, $\maxflow{G}(s,t)$ equals the maximum number of edge-disjoint
$s$-$t$-paths in~$G$.
We call this problem setting \rdel{{\sc MCPS$_1$}}\rnew{\prob{MCPS1}} to distinguish
it from the more general setting of non-uniform edge capacities\rdel{.}\rnew{:
\problemshort{MCPS1}
{The special case of \prob{MCPS} where all edge capacities are 1.}
}

\subsection{Related Work}
Capacity-preserving subgraphs
are related to the research field of sparsification.
There, given a graph~$G$, one is typically interested in an upper bound on the
size of a graph~$H$ that preserves some of~$G$'s properties up to an error
margin~$\epsilon$.
Graph~$H$ may be a minor of~$G$, a subgraph of~$G$, or a completely new graph on a
subset of~$G$'s vertices (in which case it is called a vertex
sparsifier~\cite{DBLP:conf/focs/Moitra09}).
Such research does not necessarily yield approximation algorithms
w.r.t.\ minimum sparsifier size
as the
obtained upper bound may not be easily correlated to the instance-specific
smallest possible~$H$ (e.g., the general upper bound may be $|E(H)| =
\bigO(\frac{|V|\log|V|}{\epsilon^2})$ whereas there are instances where an
optimal~$H$ is a path); however, sparsifiers often can be used as a black box in
other approximation algorithms.
A majority of cut/flow sparsification research only concerns undirected
graphs:
Bencz{\'{u}}r and Karger show how to create a subgraph that
approximately preserves the value of every cut by sampling edges of the
original graph~\cite{DBLP:conf/stoc/BenczurK96,DBLP:journals/siamcomp/BenczurK15}.
Spielman and Teng present a more general result concerning spectral graph
properties~\cite{DBLP:journals/siamcomp/SpielmanT11}.
Furthermore, there exist techniques for finding vertex sparsifiers that
preserve the congestion of any multi-commodity
flow~\cite{DBLP:conf/soda/AndoniGK14,DBLP:journals/siamcomp/EnglertGKRTT14}.
The main results for \emph{directed} graphs concern sparsifiers that depend on cut
balance~\cite{DBLP:conf/icalp/CenCP021}, and spectral sparsifiers of
strongly connected
graphs~\cite{DBLP:conf/focs/CohenKKPPRS18,DBLP:conf/stoc/CohenKPPRSV17} or
general directed graphs~\cite{DBLP:conf/icess/ZhangZF19}; however, they do not necessarily preserve cut values.

\begin{figure}[t]
    \begin{subfigure}[t]{.48\textwidth}
    \begin{center}
        \begin{tikzpicture}[scale=1.0]
                \begin{scope}[every node/.style={}]
                    \node[dot] (a) at (110:1) {};
                    \node[dot] (b) at (150:1) {};
                    \node[dot] (c) at (190:1) {};
                    \node[dot] (d) at (230:1) {};
                    \node (e) at (270:1) {\dots};
                    \node[dot] (f) at (310:1) {};
                    \node[dot] (g) at (350:1) {};
                    \node[dot] (h) at (30:1) {};
                    \node[dot] (i) at (70:1) {};
                \end{scope}

                \begin{scope}[every edge/.style={draw=myred,very thick,bend right=10}]
                    \path (a) edge (b);
                    \path (b) edge (c);
                    \path (c) edge (d);
                    \path (d) edge (e);
                    \path (e) edge (f);
                    \path (f) edge (g);
                    \path (g) edge (h);
                    \path (h) edge (i);
                    \path (i) edge[draw=black] (a);
                \end{scope}
        \end{tikzpicture}
    \end{center}
    \caption{$G$ is an undirected cycle. The red path (containing all edges but one) has stretch $|V|-1$ but
    retention ratio $\frac{1}{2}$.}
    \end{subfigure}
    \hfill
    \begin{subfigure}[t]{.48\textwidth}
    \begin{center}
        \begin{tikzpicture}[scale=1, transform shape]
                \begin{scope}[every node/.style={}]
                    \node[dot] (x) at (0,0) {};
                    \node[dot] (a) at (126:1) {};
                    \node[dot] (b) at (198:1) {};
                    \node (c) at (270:1) {\dots};
                    \node[dot] (d) at (342:1) {};
                    \node[dot] (e) at (54:1) {};
                \end{scope}

                \begin{scope}[every edge/.style={draw=black,very thick}]
                    \path (a) edge[bend right=25] (b);
                    \path (b) edge[bend right=20] (c);
                    \path (c) edge[bend right=20] (d);
                    \path (d) edge[bend right=25] (e);
                    \path (e) edge[bend right=25] (a);

                    \path (a) edge (c);
                    \path (a) edge (d);
                    \path (b) edge (d);
                    \path (b) edge (e);
                    \path (c) edge (e);
                \end{scope}

                \begin{scope}[every edge/.style={draw=myred,very thick}]
                    \path (x) edge (a);
                    \path (x) edge (b);
                    \path (x) edge (c);
                    \path (x) edge (d);
                    \path (x) edge (e);
                \end{scope}
        \end{tikzpicture}
    \end{center}
    \caption{$G$ is an undirected complete graph. The red star has stretch~$2$
        but retention ratio~$\frac{1}{|V|-1}$.}
    \end{subfigure}
    \caption{Two examples of subgraphs (in red) whose stretch
    differs greatly from their retention ratio. All edge lengths and edge
    capacities are 1, making clear that using the reciprocals of the edge
    lengths as edge capacities does not lead to a direct relation between
    stretch and capacity either.}
\label{fig:spanner_cps_comparison}
\end{figure}
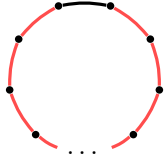
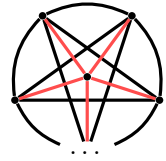

Closely related to sparsifiers are
spanners (see, e.g.,~\cite{DBLP:journals/csr/AhmedBSHJKS20} for a survey on this rich field).
These are subgraphs that preserve the length of a
shortest path within a given ratio (stretch factor) between each pair of
vertices.
However, even the most basic results in this line of work cannot be applied to
\prob{MCPS} due to fundamental differences between shortest paths and minimum
cuts (\Cref{fig:spanner_cps_comparison} illustrates this point).
Results by Räcke~\cite{DBLP:conf/stoc/Racke08}, later
generalized in~\cite{DBLP:journals/corr/abs-0907-3631} and
used for flow sparsification in~\cite{DBLP:journals/siamcomp/EnglertGKRTT14},
show (on undirected graphs) a direct correspondence between the existence of
probabilistic mappings with stretch at most~$\varphi \geq 1$ and those with
\emph{congestion} at most~$\varphi$.
However, the notion of congestion differs greatly from capacity, where the
former is defined as the maximum ratio between the flow routed over an edge in a
multi-commodity flow setting and its~capacity.

Now consider \prob{MCPS1}, where all edge capacities are~1. There, the capacity of a vertex pair is equal to the number of edge-disjoint paths between them: \prob{MCPS1}
is a special case of the \probf{DSND} problem, where one asks for the smallest
subgraph of a directed graph that satisfies given edge-connectivity requirements
for each vertex pair.
Dahl~\cite{DBLP:journals/telsys/Dahl93,DBLP:journals/dam/Dahl93} studied
\prob{DSND} from a polyhedral point of view and presented an ILP approach that
can easily be adapted to even solve the general \prob{MCPS}\rnew{ problem}.
But algorithmically, \prob{DSND} has not received as much attention as its
undirected counterpart~\cite{DBLP:journals/networks/KerivinM05}. For the latter,
Jain developed a %
2-approximation
algorithm~\cite{DBLP:journals/combinatorica/Jain01}.
Similar 2- and 4-approximation algorithms for \prob{DSND} do
exist~\cite{DBLP:journals/networks/MelkonianT04}, but these only work in a
restricted setting where the edge-connectivity requirement function is \enquote{crossing
supermodular}---which the function $\stretcha\cdot\maxflow{G}(s,t)$, required to model \prob{MCPS1} and \prob{MCPS}, is not.

Lastly, \prob{MCPS1} can be seen as a generalization of the well-established
\probf{MED} problem~\cite{DBLP:journals/siamcomp/KhullerRF95,DBLP:conf/soda/Vetta01,DBLP:journals/siamcomp/AhoGU72}:
\rdel{Given a directed graph $G=(V,E)$, one asks for a cardinality-wise minimum edge set $E' \subseteq E$
 such that the subgraph $G' = (V,E')$ preserves the
reachability relation for every pair of vertices.}
\rnew{\problem{MED}
{Digraph
$G=(V,E)$.}
{Subgraph\,$G'=(V,E')$, $E'\subseteq E$, such that, for all $(s,t) \in V^2$,
$t$ is reachable from $s$ in $G'$ if and only if $t$ is reachable from $s$ in $G$.
}{$|E'|$}%
In short, a \ac{MED} is a minimum subdigraph that preserves the reachability relation for every pair of vertices.
}%
We may think of \rdel{the}\rnew{an} \ac{MED} as a directed version of \rdel{the \acl{MST}}\rnew{a minimum spanning tree} (despite not being tree-like)---the
latter contains an undirected path from each vertex to every other reachable
vertex, the former contains a directed one.
\prob{MED} has been shown to be NP-hard via a reduction from
\probl{DHC}~\cite{DBLP:journals/siamcomp/Sahni74,DBLP:books/fm/GareyJ79,DBLP:conf/coco/Karp72}.
The NP-hardness of \prob{MCPS1} (and \prob{MCPS}) follows from a~simple~observation:

\begin{observation}
    \prob{MED} is the special case of \prob{MCPS1} with $\stretcha = \min_{(s,t)
    \in V^2} \nicefrac{1}{\maxflow{G}(s,t)}$.
\end{observation}

There exist several polynomial approximation algorithms for \prob{MED}, which are all based
on the contraction of cycles~\cite{DBLP:journals/siamcomp/KhullerRF95,DBLP:journals/ipl/ZhaoNI03};
the currently best ratio is~1.5~\cite{DBLP:conf/wads/BermanDK09,DBLP:conf/soda/Vetta01}.
Moreover, \prob{MED} can be solved optimally in linear time
on graphs whose \emph{shadow}---the underlying undirected graph obtained by ignoring edge
orientations---is series-parallel~\cite{DBLP:journals/networks/RicheyPR85}, and in
quadratic time on \acp{DAG}~\cite{DBLP:journals/siamcomp/AhoGU72}.
The latter algorithm simply deletes all those edges $uv$ for which there exists another
$u$-$v$-path.

\subsection{Our Contribution}
In this paper, we introduce the natural
problem \prob{MCPS}, which we assume may be of wider interest to the algorithmic
community.

Based on the fact that \prob{MED} is simple on \acp{DAG}, one might expect to similarly find a
polynomial algorithm for \prob{MCPS}, or at least for \prob{MCPS1}, on \acp{DAG} as well.
However, in \Cref{sec:inapproximability} we show that already on \acp{DAG} the arguably most complex core of \prob{MCPS1} cannot even be approximated within a sublogarithmic factor, unless P=NP.
Further, we prove that \prob{MCPS1} is para-NP-hard (i.e., not in the complexity classes FPT or XP, unless P=NP) w.r.t.\ several natural parameters.

In \Cref{sec:lsps} we introduce the class of \emph{\acfp{LSP}}---a
generalization of \acfp{DSP} that also allows, e.g., cycles and dense subgraphs.
\acp{LSP} have the potential to allow
simple algorithms and proofs for a wider array of problems, not just \prob{MCPS}, due to
their structural relation to \acp{DSP}.
For example,
the \prob{MCPS}-analogue of a well-known spanner
property~\cite{DBLP:journals/dcg/AlthoferDDJS93}
holds on \acp{LSP} but not on general graphs:
if the retention constraint is satisfied for all edges, it is also satisfied for
every vertex pair (see~\Cref{th:series_parallel_edge_feasible_cap_spanner}).

In the section thereafter, we complement the hardness result by several algorithmic results:
In \Cref{sec:algorithm_dsp} we present
a simple linear-time algorithm for \prob{MCPS1} on \acp{DSP}
and
a cubic-time dynamic programming algorithm for general \prob{MCPS} on \acp{DSP}.
In \Cref{sec:algorithm_lsp} we give a polynomial-time algorithm to solve \prob{MCPS} on \acp{LSP}, using a natural LSP-decomposition into several \acp{DSP} and the algorithm from the previous subsection as a blackbox. We
also propose a direct quadratic-time algorithm for \prob{MCPS1} on \acp{LSP}.
Lastly, in \Cref{sec:algorithm_other}, we demonstrate further uses of \acp{LSP}: our algorithms can directly be applied to other related problems, and we prove that the algorithm for \prob{MED} on \acp{DAG} described
in~\cite{DBLP:journals/siamcomp/AhoGU72} in fact also works on general
\acp{LSP}.

\section{Inapproximability on DAGs \& Para-NP-Hardness}
\label{sec:inapproximability}

Since a capacity-preserving subgraph always contains at least as many edges as an \ac{MED} (which might
be quite large), \prob{MCPS} can be approximated on sparse graphs by
simply returning an arbitrary feasible solution (i.e., a set of edges such that
the corresponding subgraph satisfies the capacity requirement for every vertex
pair):

\begin{observation}\label{th:cps_approximation_naive}
    Every feasible solution for a connected \prob{MCPS} instance is an \rdel{$\nicefrac{m}{n-1}$}\rnew{$\nicefrac{m}{(n-1)}$}-approximation.
\end{observation}
\begin{proof}
    Every feasible \prob{MCPS} solution must contain at least as many edges as
    an \ac{MED} to ensure a non-zero capacity for all vertex pairs~$(u,v)$
    where $v$ is reachable from $u$.
    An \ac{MED} of a connected graph is also connected.
    Thus, an optimal \prob{MCPS} solution contains at least $n-1$ edges whereas a feasible
    one contains at most all~$m$~original~edges.
\end{proof}

Hence, it seems sensible to consider a slightly altered version \prob{MCPS*} of
\rdel{the} \prob{MCPS} with a tweaked objective function $|E'| - \rdel{m_\text{MED}}\rnew{\medval}$, which does
not take into account the number of edges $\rdel{m_\text{MED}}\rnew{\medval}$ in an \ac{MED}
but aims at focusing on the problem's core complexity beyond \prob{MED}\rdel{.}\rnew{:}
\rnew{\problemdouble{MCPS*}
{Digraph
$G=(V,E)$, edge capacities~$\ecap \colon E \to \R^+$, $\stretcha \in (0,1)$.}
{Subgraph~$G'=(V,E')$, $E'\subseteq E$, such that\\
\begin{minipage}{\linewidth}\centering $\maxflow{G'}(s,t) \geq \stretcha \cdot \maxflow{G}(s,t)$ for all $(s,t) \in V^2$.\end{minipage}
}{$\varphi \be |E'| - \medval$}
{Let \prob{MCPS*1} be the special case of \prob{MCPS*} where all edge capacities are 1.}
}
\rdel{Let \prob{MCPS*1} be the special case of \prob{MCPS*} with unit edge
capacities. }%
We show that it is NP-hard to approximate \prob{MCPS*1} on \acp{DAG} to within a
sublogarithmic factor using a reduction from the decision variant of \probf{SC}:
\rdel{given a universe $\univ$ and a family of sets $\sets = \{S_i \subseteq \univ \ssep i =
1,\dots,k\}$ with $k \in \bigO(\text{poly}(|\univ|))$, one asks for a
subfamily $\setcover \subseteq \sets$ of minimum cardinality $|\setcover|$ such that
$\Union_{S \in \setcover} S = \univ$.}
\rnew{\problem{SC}
{Universe $\univ$, set family $\sets = \{S_i\}_{i\in
1,\dots,k}$ with $k \in \bigO(\text{poly}(|\univ|))$.}
{Subfamily $\setcover \subseteq \sets$ such that
$\Union_{S \in \setcover} S = \univ$.
}{$|\setcover|$}}%
For an \prob{SC} instance $(\univ, \sets)$, let $\freq(u) \be |\{S \in \sets
\ssep S \ni u\}|$ denote $u$'s
\emph{frequency}, i.e., the number of sets that contain $u$,
and $\freq \be \max_{u \in \univ} \freq(u)$ the maximum frequency.

\begin{theorem}\label{th:cpsx_1_2_inapproximability}
    Any polynomial algorithm can only guarantee an approximation ratio in
    $\Omega(\log|E|)$ for \prob{MCPS*1}, unless P=NP. This already holds on
    \acp{DAG} \rdel{with maximum path length~4}\rnew{where the longest path has length~4}.
\end{theorem}
\begin{proof}
    We give a reduction from \prob{SC} to \prob{MCPS*1} on
    \acp{DAG} such that any feasible solution for the new \prob{MCPS*1} instance
    can be transformed into a feasible solution for the original \prob{SC}
    instance with an equal or lower objective function value in linear time.
    The size $|E|$ of our \prob{MCPS*1} instance is linear in the size $N \in
    \bigO(|\univ| \cdot k) = \bigO(|U|^r)$
    of the \prob{SC} instance, i.e., $|E|=c\cdot|U|^r$ for some
    constants~$c$, $r$:
    if it was possible to approximate \prob{MCPS*1} on \acp{DAG} within a factor
    in $o(\log|E|) = o(\log (c\cdot|U|^r)) = o(\log|U|)$, one
    could also approximate \prob{SC} within $o(\log|U|)$.
    However, it is NP-hard to approximate \prob{SC} within a factor of
    $\epsilon\ln(|U|)$ for any positive $\epsilon <
    1$~\cite{DBLP:conf/stoc/DinurS14,DBLP:journals/toc/Moshkovitz15}.
    To create the \prob{MCPS*1} instance $(G,\stretcha)$, let
    $\stretcha \be \frac{1}{2}$ and construct $G$ as follows (see
    \Cref{fig:set_cover_to_cps_reduction} for a visualization of $G$):
    \begin{align*}
        G &\be (V,E) \mathrlap{\text{ with } V\be V_U \union V^U_\sets \union V_\sets \union V^\sets_t \union \{t\} \text{ and } E\be E_U \union E_\sets
        \union E_\greens \union E_\reds} &&\\
        V_U &\be \{v_{u} \mid \forall  u \in \univ\} &
        V_\sets &\be \{v_S \mid \forall S \in \sets\} \\
        V^U_\sets &\be \{x^u_S, y^u_S \mid \forall S \in \sets, u \in S\} &
        V^\sets_t &\be \{z_S \mid \forall S \in \sets\} \\
        E_U &\be \{v_ux^u_S, v_uy^u_S, x^u_Sv_S, y^u_Sv_S \mid \forall S \in \sets, u \in S\} &
        E_\sets &\be \{v_Sz_S,z_St \mid s \in \sets\} \\
        E_\greens &\be V_\sets \times \{t\} & %
        E_\reds &\be V_U \times \{t\} %
    \end{align*}
    As $G$ is a \ac{DAG}, its \ac{MED} is unique~\cite{DBLP:journals/siamcomp/AhoGU72}.
    This \ac{MED} is formed by $E_U \union E_\sets$
    and already covers all vertex pairs except $V_U \times \{t\}$. Let $(v_u,t)
    \in V_U \times \{t\}$: Its capacity in $G$ is $2\freq(u) + 1$
    and this pair thus requires a capacity of $\ceil{\frac{1}{2} \cdot
    (2\freq(u) + 1)} = \freq(u)+1$ in the solution.
    Since the \ac{MED} already has a $v_u$-$t$-capacity of~$\freq(u)$,
    only one additional edge is needed to satisfy the capacity requirement: either the
    \emph{corresponding red edge} $v_ut \in E_\reds$, or one of the \emph{corresponding
    green edges} $\{v_St \in E_\greens \ssep S \ni u\}$.
    After choosing a corresponding red or green edge for each item $u \in \univ$,
    the number of these edges is the value of the resulting solution.

    Given an \prob{SC} solution $\setcover$, we can construct an
    \prob{MCPS*1} solution
    $E_U \union E_\sets \union \{v_St \in E_\greens \ssep S \in \setcover\}$
    with the same value.
    Since every item is covered by the sets in $\setcover$, the constructed
    \prob{MCPS*1} solution includes at least one corresponding green edge for
    each item, ensuring its feasibility.

    To turn a feasible solution $E'$ for the \prob{MCPS*1} instance into a
    feasible solution for the original \prob{SC} instance with an equal or lower
    value, we remove all red edges $v_ut \in E_\reds$ from the \prob{MCPS*1}
    solution and replace each of them---if necessary---by one green edge
    $v_St \in E_\greens$ with $S \ni u$.
    Since the \prob{MCPS*1} solution has at least one corresponding green edge
    for each item $u \in \univ$, the resulting \prob{SC} solution $\{S \ssep
    v_St \in E_\greens \intersect E'\}$ also contains at least one covering set
    for each item.
\end{proof}

\begin{figure}[bt]
\tikzset{twopaths/.style={%
    to path={
    \pgfextra{%
        \pgfmathsetmacro{\startf}{-(#1-1)/2}
        \pgfmathsetmacro{\endf}{-\startf}
        \pgfmathsetmacro{\stepf}{\startf+1}
    }%
    let \p{mid}=($(\tikztostart)!0.7cm!(\tikztotarget)$) in%
    \foreach \i in {\startf,\stepf,...,\endf} {%
        (\tikztostart) -- ($ (\p{mid})!\i*6pt!90:(\tikztotarget) $) node[dot] {} -- (\tikztotarget)
    }
    \tikztonodes
}}}
\centering
\begin{tikzpicture}[scale=0.7]
        \begin{scope}[every node/.style={circle,minimum size=1pt,inner sep=1pt}]
            \node (a) at (-3,4) {$a$};
            \node (b) at (-1,4) {$b$};
            \node (c) at (1,4) {$c$};
            \node (d) at (3,4) {$d$};
            \node (s1) at (-2,2) {$S_1$};
            \node (s2) at (0,2) {$S_2$};
            \node (s3) at (2,2) {$S_3$};
            \node[dot] (s1_1) at (-1,1) {};
            \node[dot] (s2_1) at (0,1) {};
            \node[dot] (s3_1) at (1,1) {};
            \node (t) at (0,0) {$t$};
        \end{scope}

        \begin{scope}[every edge/.style={draw=black, very thick}]
            \path (a) edge[twopaths=2] (s1);
            \path (b) edge[twopaths=2] (s1);
            \path (b) edge[twopaths=2] (s3);
            \path (c) edge[twopaths=2] (s1);
            \path (c) edge[twopaths=2] (s2);
            \path (c) edge[twopaths=2] (s3);
            \path (d) edge[twopaths=2] (s2);
            \path (s1) edge (s1_1); \path (s1_1) edge (t);
            \path (s2) edge (s2_1); \path (s2_1) edge (t);
            \path (s3) edge (s3_1); \path (s3_1) edge (t);
        \end{scope}

        \begin{scope}[every path/.style={draw=mygreen, very thick,densely dotted,bend left=30}]
                \path (s1) edge (t);
                \path (s2) edge (t);
                \path (s3) edge (t);
        \end{scope}

        \begin{scope}[every path/.style={draw=myred, very thick,densely dashed,rounded corners}]
                \draw[very thick,bend right=75] (b.120) .. controls (-6,7.0) and (-4,-1) .. (t.210);
                \draw[very thick,bend right=75] (c.60) .. controls (6,7.0) and (4,-1) .. (t.-30);
                \path[bend right=50] (a) edge[-] (t);
                \path[bend left=50] (d) edge[-] (t);
        \end{scope}
\end{tikzpicture}
\caption{\prob{MCPS*1} instance constructed from the \prob{SC} instance
    $(\univ,\sets)$ with $\univ = \{a,b,c,d\}$,
    $\sets=\{\{a,b,c\},\linebreak[1]\ \{c,d\},\ \{b,c\}\}$.
    An optimal solution contains the \ac{MED} (shown as solid black lines) as well as
    one corresponding red (dashed) or green (dotted) edge for each~$u \in \univ$.
    Edges are directed from upper to lower vertices.}
    \label{fig:set_cover_to_cps_reduction}
\end{figure}
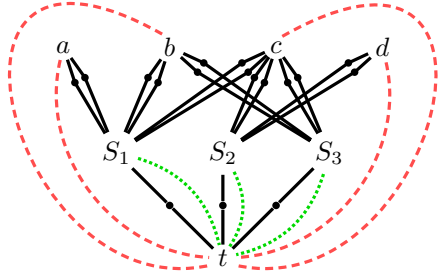

Clearly, the above reduction in particular also shows NP-hardness of \prob{MCPS1} on \acp{DAG}:
an optimal \prob{SC} solution $\{S \ssep v_St \in E_G \intersect E'\}$ can
be easily obtained from an optimal solution~$E'$ for the
\prob{MCPS1}~instance~$(G,\stretcha)$, $\stretcha=\nicefrac{1}{2}$.

\begin{corollary}
    \prob{MCPS1}, and thus also \prob{MCPS}, is NP-hard already on \acp{DAG}.
\end{corollary}

Further, the NP-hard \prob{MCPS1} instances created in our reduction are in fact very restricted \acp{DAG}.
They reveal that \prob{MCPS1} is \emph{para-NP-hard} with
respect to several parameters~$\psi$ given below, i.e., still NP-hard even when
the set of instances is restricted to those where~$\psi$ is bounded.
Consequently, \prob{MCPS1} and \prob{MCPS} are
neither in the complexity class XP (\enquote{slicewise polynomial}) nor in its
subclass FPT (\enquote{fixed-parameter tractable}) w.r.t.\ these~$\psi$
(see~\cite{DBLP:series/txtcs/FlumG06} for a primer on these complexity classes).
We refrain from giving proper definitions for all of the parameters below as this would be out-of-scope; please refer to the respective citations and the in-depth
digraph parameter comparison of~\cite{DBLP:phd/dnb/Rehs22} for more information.

\begin{corollary}
    \prob{MCPS1}, and thus also \prob{MCPS}, is para-NP-hard w.r.t.\ the
    following parameters:
    \begin{enumerate}
        \item cycle rank, %
            directed path-width, directed tree-width, DAG-width, Kelly-width,
            D-width~\cite{DBLP:journals/jct/GanianHK0ORS16},
        \item directed cut-width~\cite{DBLP:journals/jct/ChudnovskyFS12},
        \item directed feedback vertex set number and directed feedback arc set number~\cite{DBLP:conf/coco/Karp72},
        \item maximum directed path length and DAG-depth as
            defined by Ganian et al.~\cite{DBLP:journals/dam/GanianHKLOR14},
        \item maximum capacity $\max_{(u,v) \in V(G)^2} \maxflow{G}(u,v)$,
        \item K-width (i.e., the maximum number of distinct, but not necessarily
            disjoint $s$-$t$-paths over all pairs of vertices $s,t \in V(G),
            s\neq t$)~\cite{DBLP:journals/dam/GanianHKLOR14},
        \item maximum indegree and maximum outdegree,
        \item all parameters that are bounded on planar graphs, and
        \item $\psi$ such that an optimal \prob{MCPS1} solution has $n+\psi$ edges.
    \end{enumerate}
\end{corollary}
\begin{proof}
    The parameters~1--3 are bounded on
    \acp{DAG}~\cite{DBLP:journals/jct/GanianHK0ORS16, DBLP:phd/dnb/Rehs22}.
    Considering the parameters~4, the instances constructed in our reduction have maximum path length~4, and the maximum path length is an upper bound on the
    DAG-depth as defined by Ganian et al.~\cite{DBLP:journals/dam/GanianHKLOR14}.
    For the parameters~5 and~6, note that \probl{SC} is already NP-hard for
    the frequency~$\freq=2$ (in the form of \probl{VC}
    \cite{DBLP:books/fm/GareyJ79,DBLP:conf/coco/Karp72}), and observe that the
    largest capacity between any two vertices in~$G$ is~$2\freq+1$, and the
    K-width is at most~$4\freq+1$.

    Lastly, for the parameters~7--9, consider a strongly connected planar
    graph~$G$ where the indegree and outdegree of every vertex is bounded by 2.
    An optimal solution for the \prob{MCPS1} instance $(G,\stretcha)$
    with~$\stretcha = \min_{(s,t) \in V^2} \nicefrac{1}{\maxflow{G}(s,t)}$ would
    be a minimum subgraph of $G$ that has a positive $s$-$t$-capacity, and thus
    contains an $s$-$t$-path, for all $s,t \in V, s\neq t$.
    Clearly, this optimal solution would be a directed Hamilton cycle if one exists.
    However, identifying the latter is NP-hard, even under the given restrictions on
    $G$~\cite{DBLP:journals/ipl/Plesnik79}.
    In light of parameter~9, also note that an optimal \prob{MCPS1} solution for
    the given instance could only contain exactly~$n$ edges if it is a directed Hamilton cycle.
\end{proof}

Lastly, the reduction used by \Cref{th:cpsx_1_2_inapproximability} for \prob{MCPS*1} with $\stretcha = \frac{1}{2}$
can be generalized to \prob{MCPS*1} for every
$\stretcha = \frac{p}{p+1}$ with $p \in \N_{> 0}$.
This only requires a small change in the construction:
$E_U$ must contain $p+1$ $v_u$-$v_S$-paths of length 2 for all
$(v_u,v_S) \in V_U \times V_\sets$, and $E_\sets$ must contain $p$
$v_S$-$t$-paths of length~2~for~all~$v_S \in V_\sets$.

\section{Laminar Series-Parallel Graphs}
\label{sec:lsps}
In this section, we introduce \acfp{LSP}---a rich graph family that not only
includes \acfp{DSP} but also cyclic graphs and graphs with multiple
sources and sinks.
A (directed) graph $G$ is \emph{(directed) $s$-$t$-series-parallel ($s$-$t$-(D)SP)}
if and only if it is a single edge~$st$ or there exist two (directed, resp.)
$s_i$-$t_i$-series-parallel graphs $G_i$, $i \in \{1,2\}$, such that $G$
can be created from their disjoint union by one of the following operations~\cite{Duffin1965}:
\begin{enumerate}
    \item P-composition: Identify $s_1$ with $s_2$ and $t_1$ with $t_2$.
        Then, $s = s_1$ and $t = t_1$.
    \item S-composition: Identify $t_1$ with $s_2$.
        Then, $s = s_1$ and $t= t_2$.
\end{enumerate}
There also exists a widely known forbidden subgraph characterization of \acp{DSP}\rnew{, which relies on the notion of a \emph{subdivision}}:
\rnew{Let an \emph{edge subdivision operation} be the deletion of an edge~$uw$ and its replacement by a new subdivision vertex~$v$ and two new directed edges~$uv,vw$. A \emph{subdivision of a graph}~$H$ is any graph that can be created by performing a sequence of edge subdivision operations, starting with~$H$. The central ingredient of the following characterization is the characterizing digraph chiefly called $W$, depicted in \Cref{fig:graph_w}.}

\begin{theorem}[see \cite{DBLP:journals/siamcomp/ValdesTL82}]\label{th:series_parallel_forbidden_graph}
    A digraph $G=(V,E)$ is a \ac{DSP} if and only if it
    is acyclic, has exactly one source, exactly one sink, and $G$ does not
    contain a subgraph \rdel{homeomorphic to}\rnew{that is a subdivision of the digraph}~$W$ (displayed in
    \rdel{\Cref{fig:dsp_forbidden}}\rnew{\Cref{fig:graph_w}})\rdel{, i.e., a subgraph that is a subdivision of~$W$}.
\end{theorem}

\begin{figure}[p]
\begin{subfigure}[t]{.385\textwidth}
\begin{center}
\begin{tikzpicture}[scale=1.5]
        \begin{scope}[every node/.style={dot}]
            \node (a) at (0,0) {};
            \node (b) at (1,-0.5) {};
            \node (c) at (1,0.5) {};
            \node (d) at (2,0) {};
        \end{scope}
        \node[label, left=0mm of a] (la) {$x$};
        \node[label, below=0mm of b] (lb) {$z_1$};
        \node[label, above=0mm of c] (lc) {$z_2$};
        \node[label, right=0mm of d] (ld) {$y$};

        \begin{scope}[every edge/.style={edge}]
            \path (a) edge (b);
            \path (a) edge (c);
            \path (b) edge (d);
            \path (c) edge (d);
            \path (b) edge (c);
        \end{scope}
\end{tikzpicture}
\end{center}
\caption{The digraph $W$}
\label{fig:graph_w}
\end{subfigure}
\begin{subfigure}[t]{.585\textwidth}
\begin{center}
\begin{tikzpicture}[scale=1.5]
        \begin{scope}[every node/.style={dot}]
            \node (a) at (0,0) {};
            \node (b) at (1,-0.5) {};
            \node (c) at (1,0.5) {};
            \node (d) at (2,0) {};
            \node (ab1) at (0.5,-0.25) {};
            \node (bc1) at (1,-0.165) {};
            \node (bc2) at (1,0.165) {};
            \node (cd1) at (1.5,0.25) {};
        \end{scope}
        \node[label, left=0mm of a] (la) {$x$};
        \node[label, below=0mm of b] (lb) {$z_1$};
        \node[label, above=0mm of c] (lc) {$z_2$};
        \node[label, right=0mm of d] (ld) {$y$};

        \begin{scope}[every edge/.style={edge}]
            \path (a) edge (ab1);
            \path (ab1) edge (b);
            \path (b) edge (bc1);
            \path (bc1) edge (bc2);
            \path (bc2) edge (c);
            \path (a) edge (c);
            \path (b) edge (d);
            \path (c) edge (cd1);
            \path (cd1) edge (d);
        \end{scope}
\end{tikzpicture}
\end{center}
\caption{\rnew{An example of a subdivision of the digraph $W$}}
\label{fig:graph_w_subdivision}
\end{subfigure}
\caption{The digraph $W$, whose subdivisions cannot be contained in~\acp{DSP}.}%
\label{fig:dsp_forbidden}
\end{figure}
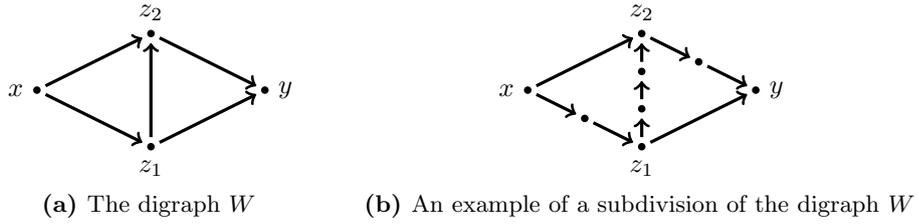
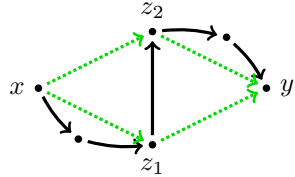
\begin{figure}
\begin{center}
\begin{tikzpicture}[scale=1.5]
        \begin{scope}[every node/.style={dot}]
            \node (a) at (0,0) {};
            \node (b) at (1,-0.5) {};
            \node (c) at (1,0.5) {};
            \node (d) at (2,0) {};
            \node (e) at (0.35,-0.45) {};
            \node (f) at (1.65,0.45) {};
        \end{scope}
        \node[label, left=0mm of a] (la) {$x$};
        \node[label, below=0mm of b] (lb) {$z_1$};
        \node[label, above=0mm of c] (lc) {$z_2$};
        \node[label, right=0mm of d] (ld) {$y$};

        \begin{scope}[every node/.style={fill=white,rectangle},
            every edge/.style={draw=black, edge}]
                \path (b) edge (c);
                \path (a) edge[bend right=10] (e);
                \path (e) edge[bend right=10] (b);
                \path (c) edge[bend left=10] (f);
                \path (f) edge[bend left=10] (d);
        \end{scope}

        \begin{scope}[every node/.style={fill=white,rectangle},
            every edge/.style={edge, draw=mygreen, densely dotted}]
            \path (a) edge (b);
            \path (b) edge (d);
            \path (a) edge (c);
            \path (c) edge (d);
        \end{scope}
\end{tikzpicture}
\end{center}
\caption{The digraph~$W$ with two added paths of length 2, see \Cref{th:obs_cps_spshadows}. Assume unit edge capacities.
Then, the $x$-$y$-capacity is 3.
For $\stretcha = \frac{1}{2}$, all edges of the original graph are covered by the
\ac{MED} (solid black edges) but the non-adjacent vertex pair~$(x,y)$ is not.
Observe that the graph is not a \ac{DSP} but its shadow is $z_1$-$z_2$-series-parallel ($(x,y) \neq (z_1,z_2)$).}
\label{fig:graph_w_modification}
\end{figure}
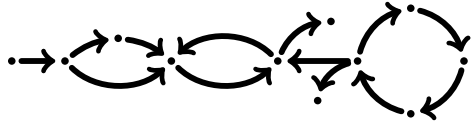
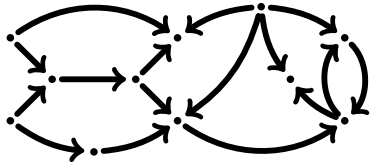
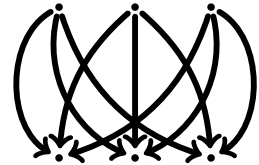
\begin{figure}[p]
    \captionsetup[subfigure]{justification=centering}
    \centering
    \begin{subfigure}[b]{\textwidth}
    \begin{center}
        \begin{tikzpicture}[scale=0.7]
                \begin{scope}[every node/.style={dot}]
                    \node (a) at (1,1.43) {};
                    \node (b) at (2,1) {};
                    \node (c) at (4,1) {};
                    \node (c1) at (5.5,1) {};
                    \node (d) at (7.5,1) {};
                    \node (e) at (6.5,0) {};
                    \node (f) at (6.5,2) {};
                    \node (g) at (0,1) {};
                    \node (g2) at (-1,1) {};
                    \node (h1) at (5,1.75) {};
                    \node (h2) at (4.75,0.25) {};
                \end{scope}

                \begin{scope}[every edge/.style={dsp}]
                    \path (a) edge[bend left=15] (b);
                    \path (c) edge[bend right=45] (b);
                    \path (b) edge[bend right=45] (c);
                    \path (c1) edge (c);
                    \path (c1) edge[bend left=30] (f);
                    \path (f) edge[bend left=30] (d);
                    \path (d) edge[bend left=30] (e);
                    \path (e) edge[bend left=30] (c1);
                    \path (g) edge[bend right=45] (b);
                    \path (g) edge[bend left=15] (a);
                    \path (g2) edge (g);
                    \path (c) edge[bend left=30] (h1);
                    \path (c1) edge[bend right=30] (h2);
                \end{scope}
        \end{tikzpicture}
    \end{center}
    \caption{Graph whose biconnected components are all \acp{DSP} or cyclic \acp{DSP}}
    \end{subfigure}\\[0.5cm]
    \begin{subfigure}[t]{.53\textwidth}
    \begin{center}
        \begin{tikzpicture}[scale=0.55]
                \begin{scope}[every node/.style={dot}]
                    \node (a) at (0,0) {};
                    \node (b) at (6,0.75) {};
                    \node (c) at (4,0) {};
                    \node (d) at (1,-1) {};
                    \node (e) at (3,-1) {};
                    \node (f) at (0,-2) {};
                    \node (g) at (2,-2.75) {};
                    \node (h) at (4,-2) {};
                    \node (i) at (8,0) {};
                    \node (k) at (8,-2) {};
                    \node (l) at (6.7,-1) {};
                \end{scope}

                \begin{scope}[every edge/.style={dsp}]
                    \path (a) edge[bend left=35] (c);
                    \path (a) edge (d);
                    \path (d) edge (e);
                    \path (e) edge (c);
                    \path (f) edge (d);
                    \path (e) edge (h);
                    \path (f) edge[bend right=15] (g);
                    \path (g) edge[bend right=15] (h);
                    \path (b) edge[bend right=15] (c);
                    \path (b) edge[bend left=15] (i);
                    \path (h) edge[bend right=35] (k);
                    \path (i) edge[bend left=45] (k);
                    \path (k) edge[bend left=45] (i);
                    \path (b) edge[bend left=20] (h);
                    \path (b) edge[bend right=15] (l);
                    \path (k) edge[bend left=15] (l);
                \end{scope}
        \end{tikzpicture}
    \end{center}
    \caption{A more complex biconnected \ac{LSP}}
    \end{subfigure}\hfill
    \begin{subfigure}[t]{.43\textwidth}
    \begin{center}
        \begin{tikzpicture}[scale=1.0, transform shape, rotate=-90]
                \begin{scope}[every node/.style={dot}]
                    \node (a) at (0,0) {};
                    \node (b) at (0,1) {};
                    \node (c) at (0,2) {};
                    \node (d) at (2,0) {};
                    \node (e) at (2,1) {};
                    \node (f) at (2,2) {};
                \end{scope}

                \begin{scope}[every edge/.style={dsp}]
                    \path (a) edge[bend right=60] (d);
                    \path (a) edge[bend right=40] (e);
                    \path (a) edge[bend right=27] (f);
                    \path (b) edge[bend right=22] (d);
                    \path (b) edge (e);
                    \path (b) edge[bend left=22] (f);
                    \path (c) edge[bend left=27] (d);
                    \path (c) edge[bend left=40] (e);
                    \path (c) edge[bend left=60] (f);
                \end{scope}
        \end{tikzpicture}
    \end{center}
    \caption{Complete bipartite graph with natural orientation}
    \end{subfigure}
    \caption{Examples of \acp{LSP}. Every edge represents a \ac{DSP} of arbitrary size.}%
\label{fig:lsp_examples}
\end{figure}

\FloatBarrier
Given a directed graph~$G = (V,E)$, for every vertex pair $(u,v) \in V^2$,
let $\pathind{G}{u,v}$ be the graph induced by the edges on
$u$-$v$-paths.
Note that such a path-induced subgraph may contain cycles but a single path may
not.
If $e=uv$ is an edge, we call $\pathind{G}{u,v}$ an \emph{\ac{EAS}} and may
use the shorthand notation~$\pathind{G}{e}$.
Based on these notions, we can define \acp{LSP}:

\begin{definition}[Laminar Series-Parallel Graph]
    A directed graph $G=(V,E)$ is a \emph{\acf{LSP}} if and only if it
    satisfies:
    \begin{prop}\label{th:lsp_psp}
        For every $(s,t) \in V^2$, $\pathind{G}{s,t}$
        is either an $s$-$t$-\ac{DSP} or contains no edges; and
    \end{prop}
    \begin{prop}\label{th:lsp_laminar}
        $\{E(\pathind{G}{e})\}_{e \in E}$ form a laminar set family, i.e.,
        for all edges $e_1,e_2 \in E$ we~have
        \begin{equation*}
            \pathind{G}{e_1} \subseteq \pathind{G}{e_2}~\lor~\pathind{G}{e_2} \subseteq \pathind{G}{e_1}~\lor~E(\pathind{G}{e_1}) \intersect E(\pathind{G}{e_2}) = \emptyset.
        \end{equation*}
    \end{prop}
\end{definition}

\Cref{fig:lsp_examples} shows some example \acp{LSP}.
\acp{LSP} not only include the graphs whose biconnected components are all
\acp{DSP}, but also some cyclic graphs, e.g., \emph{cyclic \acp{DSP}}
constructed by identifying the source and sink of a \ac{DSP}.
Moreover, there exist very dense \acp{LSP}, e.g.,
the natural orientations of complete bipartite graphs.

We present some interesting properties of \acp{LSP}.
Graphs that only need to satisfy property \Cref{th:lsp_psp} or \Cref{th:lsp_laminar} may be called \Cref{th:lsp_psp}- or \Cref{th:lsp_laminar}-graphs, respectively.
First, we show that \Cref{th:lsp_psp}-graphs can be characterized as those graphs that do not contain a \rdel{$W$-subdivision}\rnew{subdivision of the digraph~$W$}:

\begin{theorem}\label{th:psp_forbidden_graph}
    A directed graph $G=(V,E)$ satisfies \Cref{th:lsp_psp} if and only
    if $G$ does not contain a subgraph \rdel{homeomorphic to}\rnew{that is a subdivision of the digraph}~$W$ (displayed in
    \rdel{\Cref{fig:dsp_forbidden}}\rnew{\Cref{fig:graph_w}}).
\end{theorem}

\newcommand{\pathset}{\mathcal{P}}
\newcommand{\vpout}[1]{v^\lout(#1)}
\newcommand{\vpin}[1]{v^\lin(#1)}
\newcommand{\vout}{v^\lout}
\newcommand{\vin}{v^\lin}
\newcommand{\Vout}{V^\lout}
\newcommand{\Vin}{V^\lin}
\newcommand{\Oout}{O^\lout}
\newcommand{\Oin}{O^\lin}
\newcommand{\Cout}{C_b}
\newcommand{\Cin}{C_a}
\newcommand{\pref}[1]{\ensuremath{\mathit{pre}(#1)}}
\newcommand{\postf}[1]{\ensuremath{\mathit{post}(#1)}}
\newcommand{\comint}[2]{\ensuremath{#1\pitchfork#2}}
\begin{proof}
    Assume $G$ contains a \rdel{subgraph homeomorphic to}\rnew{a subdivision of}~$W$\rnew{ as a subgraph}. Let $H$ be the smallest
    such subgraph, and let $s$ and~$t$ be its source and sink, respectively. Note
    that $H \subseteq \pathind{G}{s,t}$.
    Hence, $\pathind{G}{s,t}$ cannot be a \ac{DSP}, and $G$ is not a
    \Cref{th:lsp_psp}-graph.

\begin{figure}[pt]
\begin{center}
\begin{tikzpicture}[yscale=0.825]
        \begin{scope}[every node/.style={dot}, every label/.style={above,draw,rectangle}]
            \node (1-s) at (2,15){};
            \node (1-c1) at (4,14){};
            \node (1-c2) at (4,15){};
            \node (1-c3) at (5,16){};
            \node (1-c4) at (7,16){};
            \node (1-c5) at (8,15){};
            \node (1-c6) at (8,14){};
            \node (1-c7) at (7,14){};
            \node (1-t) at (10,15){};
        \end{scope}
        \node[label, above left=1.5cm and 1.8cm of 1-s] (1-lcase) {\rlap{\textbf{Case 1} (note that~$y$ may lie anywhere after~$z_2$, but before~$x$, on~$C$):}};
        \node[label, below=0mm of 1-s] (1-lx)  {$s$};
        \node[label, below=0mm of 1-t] (1-ly) {$t$};
        \node[label, below left=1mm of 1-c1] (1-l1)  {$x$};
        \node[label, below right=-1mm and 1mm of 1-c2] (1-l2) {$\vpin{P}$};
        \node[label, above left=-0.3mm and 0.7mm of 1-c3] (1-l3)  {$z_1$};
        \node[label, above right=-0.3mm and 0.7mm of 1-c4] (1-l4)  {$z_2$};
        \node[label, below left=-1mm and 1mm of 1-c5] (1-l5) {$\vpout{P}$};
        \node[label, below right=1mm of 1-c6] (1-l6) {$y$};
        \node[label, below=1mm of 1-c7] (1-l7) {$\vout_\ell$};

        \begin{scope}[every edge/.style={draw=black,->,path}]
            \path (1-s) edge[] (1-c2);
            \path (1-c5) edge[] (1-t);
            \path (1-c1) edge[myred] (1-c2);
            \path (1-c2) edge[myred,dashed] (1-c3);
            \path (1-c3) edge[myred] (1-c4);
            \path (1-c4) edge[myred,dashed] (1-c5);
            \path (1-c5) edge[myred] (1-c6);
            \path (1-c6) edge[dashed] (1-c7);
            \path (1-c7) edge[] (1-c1);
            \path[->,path,draw=myred,bend left=90] (1-c3) .. controls (7.5,17.5) and (10,16) .. (1-c6);
            \path[->,path,draw=myred,bend left=90] (1-c1) .. controls (2,16) and (4.5,17.5) .. (1-c4);
        \end{scope}

        \begin{scope}[on background layer, every edge/.style={cycle-shade}]
            \path (1-c1.center) edge (1-c2.center);
            \path (1-c2.center) edge (1-c3.center);
            \path (1-c3.center) edge (1-c4.center);
            \path (1-c4.center) edge (1-c5.center);
            \path (1-c5.center) edge (1-c6.center);
            \path (1-c6.center) edge (1-c7.center);
            \path (1-c7.center) edge (1-c1.center);
        \end{scope}

        \begin{scope}[every node/.style={dot}, every label/.style={above,draw,rectangle}]
            \node (2a0-s) at (0,10){};
            \node (2a0-x) at (1,10){};
            \node (2a0-c1) at (4,10){};
            \node (2a0-c2) at (4,11){};
            \node (2a0-c3) at (8,11){};
            \node (2a0-c4) at (8,10){};
            \node (2a0-y) at (11,10){};
            \node (2a0-t) at (12,10){};
        \end{scope}
        \node[label, above=1.7cm of 2a0-s] (2a0-lcase) {\rlap{\textbf{Case 2a} --- both prefix/postfix pairs have no common internal vertices:}};
        \node[label, below=0mm of 2a0-s] (2a0-lx)  {$s$};
        \node[label, below=0mm of 2a0-x] (2a0-lx)  {$x$};
        \node[label, below=0mm of 2a0-y] (2a0-ly) {$y$};
        \node[label, below=0mm of 2a0-t] (2a0-ly) {$t$};
        \node[label, left=1mm of 2a0-c1] (2a0-l1)  {$z_1=\vpout{P_2}$};
        \node[label, above left=-1mm and 1mm of 2a0-c2] (2a0-l2) {$z_2=\vpin{P_2}$};
        \node[label, above right=-1mm and 1mm of 2a0-c3] (2a0-l3) {$\vpout{P_1}$};
        \node[label, right=1mm of 2a0-c4] (2a0-l4) {$\vpin{P_1}$};

        \begin{scope}[every edge/.style={draw=black,->,path}]
            \path (2a0-s) edge[dashed] (2a0-x);
            \path (2a0-x) edge[myred] (2a0-c2);
            \path (2a0-x) edge[myred,bend right=35] (2a0-c4);
            \path (2a0-c1) edge[myred] (2a0-c2);
            \path (2a0-c2) edge[myred] (2a0-c3);
            \path (2a0-c3) edge[] (2a0-c4);
            \path (2a0-c4) edge[myred] (2a0-c1);
            \path (2a0-c1) edge[myred,bend right=35] (2a0-y);
            \path (2a0-c3) edge[myred] (2a0-y);
            \path (2a0-y) edge[dashed] (2a0-t);
        \end{scope}

        \begin{scope}[on background layer, every edge/.style={cycle-shade}]
            \path (2a0-c1.center) edge (2a0-c2.center);
            \path (2a0-c2.center) edge[] (2a0-c3.center);
            \path (2a0-c3.center) edge (2a0-c4.center);
            \path (2a0-c4.center) edge[] (2a0-c1.center);
        \end{scope}

        \begin{scope}[every node/.style={dot}, every label/.style={above,draw,rectangle}]
            \node (2a2-s) at (0,4.7){};
            \node (2a2-x) at (1,4.7){};
            \node (2a2-z) at (6,3.7){};
            \node (2a2-z2) at (6,6.7){};
            \node (2a2-c1) at (4,4.7){};
            \node (2a2-c2) at (4,5.7){};
            \node (2a2-c3) at (8,5.7){};
            \node (2a2-c4) at (8,4.7){};
            \node (2a2-y) at (11,4.7){};
            \node (2a2-t) at (12,4.7){};
        \end{scope}
        \node[label, above=2.2cm of 2a2-s] (2a2-lcase) {\rlap{\textbf{Case 2a} --- both prefix/postfix pairs have common internal vertices:}};
        \node[label, below=0mm of 2a2-s] (2a2-lx)  {$s$};
        \node[label, below=0mm of 2a2-x] (2a2-lx)  {$x$};
        \node[label, below=0mm of 2a2-z] (2a2-lz)  {$z_1$};
        \node[label, above=0mm of 2a2-z2] (2a2-lz2)  {$z_2$};
        \node[label, below=0mm of 2a2-y] (2a2-ly) {$y$};
        \node[label, below=0mm of 2a2-t] (2a2-ly) {$t$};
        \node[label, left=1mm of 2a2-c1] (2a2-l1)  {$\vpout{P_2}$};
        \node[label, left=1mm of 2a2-c2] (2a2-l2) {$\vpin{P_2}$};
        \node[label, right=1mm of 2a2-c3] (2a2-l3) {$\vpout{P_1}$};
        \node[label, right=1mm of 2a2-c4] (2a2-l4) {$\vpin{P_1}$};

        \begin{scope}[every edge/.style={draw=black,->,path}]
            \path (2a2-s) edge[dashed] (2a2-x);
            \path (2a2-x) edge[myred,bend left=30] (2a2-z2);
            \path (2a2-x) edge[myred,bend right=20] (2a2-z);
            \path (2a2-z) edge[myred,bend right=6] (2a2-c4);
            \path (2a2-c1) edge[myred] (2a2-c2);
            \path (2a2-c2) edge[myred] (2a2-c3);
            \path (2a2-z2) edge[bend right=6] (2a2-c2);
            \path (2a2-c3) edge[myred,bend right=6] (2a2-z2);
            \path (2a2-c3) edge[] (2a2-c4);
            \path (2a2-c4) edge[myred] (2a2-c1);
            \path (2a2-c1) edge[bend right=6] (2a2-z);
            \path (2a2-z) edge[myred,bend right=20] (2a2-y);
            \path (2a2-z2) edge[myred,bend left=30] (2a2-y);
            \path (2a2-y) edge[dashed] (2a2-t);
        \end{scope}

        \begin{scope}[on background layer, every edge/.style={cycle-shade}]
            \path (2a2-c1.center) edge (2a2-c2.center);
            \path (2a2-c2.center) edge[] (2a2-c3.center);
            \path (2a2-c3.center) edge (2a2-c4.center);
            \path (2a2-c4.center) edge[] (2a2-c1.center);
        \end{scope}

        \begin{scope}[every node/.style={dot}, every label/.style={above,draw,rectangle}]
            \node (2b-s) at (0,0){};
            \node (2b-x) at (1,0){};
            \node (2b-z) at (6,-1){};
            \node (2b-z2) at (6,2){};
            \node (2b-c1) at (6,0){};
            \node (2b-c2) at (6,1){};
            \node (2b-y) at (11,0){};
            \node (2b-t) at (12,0){};
        \end{scope}
        \node[label, above=2.1cm of 2b-s] (2b-lcase) {\rlap{\textbf{Case 2b}:}};
        \node[label, below=0mm of 2b-s] (2b-lx)  {$s$};
        \node[label, below=0mm of 2b-x] (2b-lx)  {$x$};
        \node[label, below=0mm of 2b-z] (2b-lz)  {$z_1$};
        \node[label, above=0mm of 2b-z2] (2b-lz2)  {$z_2$};
        \node[label, below=0mm of 2b-y] (2b-ly) {$y$};
        \node[label, below=0mm of 2b-t] (2b-ly) {$t$};
        \node[label] (2b-l1) at (7.55,1.55) {$\vpin{P_2}=\vpout{P_1}$};
        \node[label] (2b-l2) at (4.45,-0.55) {$\vpin{P_1}=\vpout{P_2}$};

        \begin{scope}[every edge/.style={draw=black,->,path}]
            \path (2b-s) edge[dashed] (2b-x);
            \path (2b-x) edge[myred,bend left=30] (2b-z2);
            \path (2b-x) edge[myred,bend right=30] (2b-z);
            \path (2b-z) edge[myred,dashed,bend right=30] (2b-c1);
            \path (2b-z2) edge[dashed,bend right=30] (2b-c2);
            \path (2b-c2) edge[myred,dashed,bend right=30] (2b-z2);
            \path[draw=myred, path, ->, bend left=30] (2b-c1) -- (4,0) -- (4,1) -- (2b-c2);
            \path[draw=black, path, ->, bend left=30] (2b-c2) -- (8,1) -- (8,0) -- (2b-c1);
            \path (2b-c1) edge[dashed,bend right=30] (2b-z);
            \path (2b-z) edge[myred,bend right=30] (2b-y);
            \path (2b-z2) edge[myred,bend left=30] (2b-y);
            \path (2b-y) edge[dashed] (2b-t);
        \end{scope}

        \begin{scope}[on background layer, every edge/.style={cycle-shade}]
            \path[cycle-shade] (2b-c1.center) -- (4,0) -- (4,1) -- (2b-c2.center);
            \path[cycle-shade] (2b-c2.center) -- (8,1) -- (8,0) -- (2b-c1.center);
        \end{scope}
\end{tikzpicture}
\end{center}
\caption{Visualization for the proof of \Cref{th:psp_forbidden_graph}.
    $\pathind{G}{s,t}$ contains a cycle (shaded in gray).
    The wiggly lines denote paths.
    Dashed paths may have length 0, others must contain at least one edge.
    The \rdel{$W$-subdivision}\rnew{subdivision of the digraph~$W$} is given by the vertices $x,y,z_1,z_2$~(as named
    in~\Cref{fig:dsp_forbidden}), and highlighted in red.
    For readability, the figures only show cases where a prefix/postfix pair may have at most one common vertices, even though there could exist more; by the choice of $z_1$ and $z_2$, such further common vertices are irrelevant for the \rdel{$W$-subdivision}\rnew{subdivision of~$W$}.
    }
\label{fig:path_induced_cycle_k4}
\end{figure}

    For the other direction assume that $G$ is not a \Cref{th:lsp_psp}-graph. Then it contains a vertex
    pair~$(s,t)$ such that $H' \be \pathind{G}{s,t}$ is neither empty nor a
    \ac{DSP}.
    Since $H'$ contains exactly one source and one sink,
    \Cref{th:series_parallel_forbidden_graph} lets us conclude that $H'$ (and
    thus $G$) either contains a \rdel{subgraph homeomorphic to}\rnew{subdivision of the digraph}~$W$ or a (directed) cycle~$C$.
    We will argue that in the latter case, $H'$ must contain a \rdel{subgraph
    homeomorphic to~$W$, i.e.\ a~$W$-subdivision,}\rnew{subdivision of~$W$} as well.

    So consider a set of $s$-$t$-paths~$\pathset$ that is spanning~$C$.
    For a path~$P\in \pathset$, let the \emph{entrance vertex}~$\vpin{P}$ (\emph{exit vertex}~$\vpout{P}$) be the first (last, respectively) vertex along~$P$ that $P$~and~$C$ have in common. Note that not all vertices of~$P$ between~$\vpin{P}$ and~$\vpout{P}$ have to be on~$C$.
    We first argue that there exist two different entrance vertices and two
    different exit vertices on~$C$ (not necessarily four distinct vertices):
    Since no single $s$-$t$-path visits the same vertex twice, no path $P'\in\pathset$
    contains~$C$ completely.
    There must be another $s$-$t$-path~$P'' \in \pathset$ that contains
    the cycle's edge going to $\vpin{P'}$, and $P''$ thus enters~$C$ for the
    first time at a vertex~$\vpin{P''} \neq \vpin{P'}$.
    Analogously, there must be an $s$-$t$-path~$P''' \in \pathset$ that contains
    the cycle's edge going out of $\vpout{P'}$, and~$\vpout{P'''} \neq \vpout{P'}$.

    Based on the existence of these vertices, we can find a subdivision of~$W$
    in~$H'$ (and consequently in $G$) as shown in
    \Cref{fig:path_induced_cycle_k4}.
    To establish the \rdel{$W$-}subdivision, we will assign vertices in $H'$ to the
    vertices~$x,y,z_1,z_2$ of~$W$ (the naming scheme of these vertices
    corresponds to the one in \Cref{fig:dsp_forbidden}).
    Let~$\Vin$ with~$k \be |\Vin| \geq 2$ \,($\Vout$ with~$\ell \be |\Vout| \geq
    2$) be the set of all vertices that serve as an entrance vertex (exit
    vertex, respectively) for some $s$-$t$-path in $\pathset$.
    We make a case distinction regarding the clockwise cyclic order~$O$ of all
    these vertices~$\Vin \union \Vout$ around~$C$.
    Mind the fact that a vertex can serve multiple times both as an entrance
    vertex and as an exit vertex for different paths in $\pathset$; in
    particular, $\Vin \intersect \Vout$ may not be empty.

    \begin{description}
    \item[Case 1:]
    The vertices in~$\Vin$ and the vertices in~$\Vout$
    form continuous sequences~$O^\lin, O^\lout$ in~$O$ respectively, and there
    is at most one vertex that is contained in both~$\Oin$ and~$\Oout$.
    We assume w.l.o.g.\ that either $O =
    (\vin_1,\dots,\vin_k,\vout_1,\dots,\vout_\ell)$ or $O =
    (\vin_1,\dots,\vin_k=\vout_1,\dots,\vout_\ell)$:
    The edge~$e \in C$ going out of~$\vout_\ell$ must be contained
    in a path~$P\in\pathset$.
    Let~$\Cin$~($\Cout$) be the subpath of~$C$ from~$\vpin{P}$ to~$\vpout{P}$ (from~$\vpout{P}$ to~$\vpin{P}$, respectively).
    Then, $P$ enters~$\Cin$ through the entrance vertex~$\vpin{P}$, but must leave it again in order to avoid meeting~$\vpout{P}$ prematurely. %
    Since~$P$ has to reach~$e$, %
    it eventually enters~$\Cout$
    after~$\vpout{P}$, but before or at~$\vout_\ell$. %
    After traversing~$e$, $P$ must then leave~$\Cout$ to
    avoid~$\vpin{P}$, and finally enter~$\Cin$ again to reach~$\vpout{P}$.
    Note that~$P$ may repeatedly exit and enter~$C$ apart from the mentioned occasions.
    Nonetheless, we know for certain that there exists at least one subpath of~$P$ that exits~$\Cout$ and enters~$\Cin$, sharing only its first and its final vertex with~$C$.
    Among such subpaths, consider the one whose first vertex is as close to $\vpin{P}$ along $C$ as possible; call its first and last vertex $x$ and $z_2$, respectively.
    By construction, $\vpin{P}$ lies between~$x$ and~$z_2$ on~$C$, and
    thus there exists a subpath of~$P$ that leaves~$C$ at a vertex~$z_1$ after~$\vpin{P}$ (in orientation of $C$), but before~$z_2$; this subpath enters~$C$ again at a vertex~$y$ after~$z_2$, but before~$x$ (such that~$P$ can still reach~$\vpout{P}$ later).
    Overall, $P\cup C\subseteq H'$ is a supergraph of a \rdel{$W$-subdivision}\rnew{subdivision of~$W$}, as witnessed by the labeling of the vertices $x,y,z_1,z_2$ and disregarding the relative position of these vertices w.r.t.\ $\vout_\ell$.

    \item[Case 2:]
    Either the vertices in~$\Vin$ or the vertices in~$\Vout$ do not form a
    continuous sequence in~$O$ (\emph{Case~2a}), or~$|\Vin \intersect \Vout|\geq 2$ (\emph{Case~2b}).
    In \textbf{Case~2a}, there are two entrance vertices~$\vin_1,
    \vin_2$ and two exit vertices~$\vout_1, \vout_2$ such that~$O$ alternates
    between them as follows: $O =
    (\vin_1,\dots,\vout_2,\dots,\vin_2,\dots,\vout_1,\dots)$.
    For each $i\in\{1,2\}$, let~$P_i^\lin$ ($P_i^\lout$) be a path in~$\pathset$
    with~$\vin_i$ as its entrance vertex ($\vout_i$ as its exit vertex,
    respectively); we can establish an $s$-$t$-path~$P_i$ by concatenating the subpath from~$s$ to~$\vin_i$
    of~$P^\lin_i$, the subpath from~$\vout_i$ to~$\vin_i$ on~$C$, and the
    subpath of~$\vout_i$ to~$t$ of~$P^\lout_i$.
    In \textbf{Case~2b}, we select two vertices in $\Vin \intersect \Vout$ and construct the $s$-$t$-paths~$P_1, P_2$ analogously such that $\vpin{P_1}=\vpout{P_2}$ and $\vpin{P_2}=\vpout{P_1}$.

    From now on, both Cases~2a and~2b can be discussed together as Case~2. $P_1$ and~$P_2$ have the same start and end vertex and may have several more common vertices on and outside of~$C$.
    Traversing along~$P_1$, let~$x$ be the last such common vertex before~$\vpin{P_1}$, and let~$y$ be the first such common vertex after~$\vpout{P_1}$.
    We will call the subpath of an $s$-$t$-path~$P$ from~$x$ to~$\vpin{P}$
    the \emph{prefix}~\pref{P} of~$P$, and the subpath~from~$\vpout{P}$ to~$y$ the
    \emph{postfix}~\postf{P} of~$P$.
    Let $U_1$ ($U_2$) be the common internal vertices of $\pref{P_1}$ and $\postf{P_2}$ ($\pref{P_2}$ and $\postf{P_1}$, respectively).
    Then, we choose~$z_1$ to be the last vertex of $U_1\cup\{\vpout{P_2}\}$ when traversing along~\postf{P_2}.
    Analogously, let $z_2$ be the first vertex of $U_2\cup\{\vpin{P_2}\}$  when traversing along~\pref{P_2}.
    Always, $P_1\cup P_2\cup C\subseteq H'$ is now a supergraph of a \rdel{$W$-subdivision}\rnew{subdivision of~$W$}, as witnessed by the labeling of the vertices $x,y,z_1,z_2$.
    \end{description}
\end{proof}

Unlike \Cref{th:lsp_psp}-graphs, there can be no characterization of
\Cref{th:lsp_laminar}-graphs via forbidden subdivisions or (directed versions
of) minors:%

\begin{theorem}\label{th:laminar_subdivision}
    Every directed graph $G$ has a subdivision $\subdiv{G}$ that satisfies
    \Cref{th:lsp_laminar}.
\end{theorem}
\begin{proof}
    Construct $\subdiv{G}$ by subdividing every edge of~$G$ (i.e., removing every edge~$uv$ of~$G$ and adding a new subdivision vertex~$w$ as well as directed edges $uw$, $wv$ instead).
    Each edge $e \in E(\subdiv{G})$ is either the only incoming or the only outgoing edge
    at some subdivision vertex.
    Thus, $\pathind{\subdiv{G}}{e}$ contains only edge~$e$.
    As all $\pathind{\subdiv{G}}{e}$ are disjoint,
    $\subdiv{G}$ satisfies \Cref{th:lsp_laminar}.
\end{proof}

Nevertheless, \rnew{based on \Cref{th:psp_forbidden_graph} }it is easy to see:

\begin{theorem}\label{th:dsp_lsp}
    Every \ac{DSP} $G$ is an \ac{LSP}.
\end{theorem}
\begin{proof}
    $G$ satisfies \Cref{th:lsp_psp} by \Cref{th:psp_forbidden_graph}.
    For \Cref{th:lsp_laminar}, we use induction over the directed
    series-parallel composition of $G$.
    A graph with a single edge $e$ satisfies \Cref{th:lsp_laminar} since there
    is no other \ac{EAS} that $\pathind{G}{e}$ could intersect with.
    Now consider the creation of a new $s$-$t$-\ac{DSP} $G$ from two \acp{DSP}
    $G_1$ and $G_2$, for which \Cref{th:lsp_laminar} holds.
    If $st \notin E(G)$, the composition does not create any new paths between
    the endpoints of any edge, so $\pathind{G}{e} = \pathind{G_i}{e}$ for every
    $e \in E(G_i)$, $i \in \{1,2\}$;
    we already know that $\{E(\pathind{G_i}{e})\}_{e \in E(G_i)}$ for $i \in
    \{1,2\}$ form laminar set families, respectively, so
    $\{E(\pathind{G}{e})\}_{e \in E(G_1) \union E(G_2)}$ does so as well.
    Finally, assume $st \in E(G)$.
    We only have to show that the laminar set family $\{E(\pathind{G}{e})\}_{e
    \in E(G) \setminus \{st\}}$ remains laminar when $E(\pathind{G}{s,t})$ is
    added to it.
    This is the case since \rdel{$\pathind{G}{s,t}$ fully contains $G$}\rnew{$E(\pathind{G}{s,t}) = E(G)$ is a superset of all other possible subsets of~$E(G)$; it thus serves as the top-most set in the laminar family}.
\end{proof}

\section{Efficient Algorithms}
\label{sec:algorithms}
We start with presenting two algorithms that operate on \acp{DSP}
(\Cref{sec:algorithm_dsp}): the first one finds an optimal \prob{MCPS1} solution
in linear time, and the second one an optimal \prob{MCPS} solution in
cubic time. Afterwards, we describe how to lift these results to \acp{LSP}
and give an additional concrete quadratic-time algorithm for \prob{MCPS1} on
\acp{LSP}~(\Cref{sec:algorithm_lsp}). We demonstrate other algorithmic corollaries in \Cref{sec:algorithm_other}.

\subsection{\texorpdfstring{MCPS$_1$}{MCPS1} and MCPS on Directed Series-Parallel Graphs}
\label{sec:algorithm_dsp}

Our algorithms exploit a useful property of capacities in
\Cref{th:lsp_psp}-graphs:
if every edge is covered, then \emph{all} vertex pairs are covered.

\begin{theorem}\label{th:series_parallel_edge_feasible_cap_spanner}
    Given a retention ratio $\stretcha \in (0,1)$, let $G=(V,E)$ be a
    \Cref{th:lsp_psp}-graph with edge capacities~\ecap, and
    $G'=(V,E')$, $E' \subseteq E$, a subgraph with $\maxflow{G'}(u,v) \geq \stretcha \cdot
    \maxflow{G}(u,v)$ for all edges $uv \in E$.
    Then, $E'$ is a feasible $\stretcha$-\prob{MCPS} solution, i.e.,
    $\maxflow{G'}(u,v) \geq \stretcha \cdot \maxflow{G}(u,v)$ for all vertex pairs~$(u,v) \in V^2$.
\end{theorem}

\begin{proof}
    Consider any vertex pair $(u,v) \in V^2$ with $u \neq v$ (as vertex pairs
    with $u = v$ are trivially covered).
    We can assume w.l.o.g.\ that a maximum flow from $u$ to $v$ passes only over
    edges of $H \be \pathind{G}{u,v}$:
    for any maximum flow from $u$ to $v$ that uses cycles outside of $H$,
    we can find one of the same value in $H$ by simply removing these cycles.
    As $G$ is a \Cref{th:lsp_psp}-graph, $H$ is a $u$-$v$-\ac{DSP}.
    We use induction over the series-parallel composition of~$H$ to prove
    that $(u,v)$ is covered.
    If $uv \in E$, the edge is already covered as asserted in the prerequisites
    of the theorem; this includes the base case of $H$ containing a single edge.

    Let $H$ be created from the disjoint union of two smaller
    $u_i$-$v_i$-\acp{DSP} $H_i$, $i \in \{1,2\}$, for which the theorem holds.
    Further, let $X' \be (V(X), E(X) \intersect E(G'))$ for any subgraph~$X \in
    \{H,H_1,H_2\}$ of~$G$.
    If $H$ is constructed from an S-composition---i.e.\ $u = u_1$, $v = v_2$, and $v_1 = u_2$---each maximum $u$-$v$-flow~in~$H$ ($H'$) passes through $H_1$ and $H_2$ ($H'_1$ and $H'_2$, resp.):
    \begin{equation*}
    \maxflow{H'}(u,v)
    = \min\{\maxflow{H_1'}(u,v_1),\maxflow{H_2'}(v_1,v)\}
    \geq \min\{\stretcha \cdot \maxflow{H_1}(u,v_1),\stretcha \cdot \maxflow{H_2}(v_1,v)\}
    = \stretcha \cdot \maxflow{H}(u,v).
    \end{equation*}
    If $H$ is constructed from a P-composition---i.e.\ $u = u_1 = u_2$
    and $v = v_1 = v_2$---its $u$-$v$-capacity in~$H$~($H'$) is the sum of
    $u$-$v$-capacities in~$H_1$~and~$H_2$ ($H'_1$ and $H'_2$, resp.):
    \begin{equation*}
    \maxflow{H'}(u,v)
    = \maxflow{H_1'}(u,v) + \maxflow{H_2'}(u,v)
    \geq \stretcha \cdot \maxflow{H_1}(u,v) + \stretcha \cdot \maxflow{H_2}(u,v)
    = \stretcha \cdot \maxflow{H}(u,v).
    \qedhere\popQED
    \end{equation*}
\end{proof}

\begin{remark}\label{th:obs_cps_spshadows}
    \Cref{th:series_parallel_edge_feasible_cap_spanner} does in general not
    apply to graphs for which only their shadow is
    series-parallel.
    In particular, it does not even hold for the graph~$W$---the smallest graph
    without property~\Cref{th:lsp_psp}---when two paths of length 2 are added to
    it, see \rdel{\Cref{fig:dsp_forbidden}}\rnew{\Cref{fig:graph_w_modification}}.
\end{remark}

We start with giving a simple linear-time algorithm to solve \prob{MCPS1}, the problem where we assume uniform edge capacities, in \acp{DSP}.
The algorithm requires the series-parallel decomposition tree to be \emph{clean}, i.e.,
if there are multiple P-compositions of several $s$-$t$-\acp{DSP} $H_0,\dots,H_k$
where $E(H_0) =\{st\}$ is a single edge, we first compose $H_1,\dots,H_k$ into a common
$s$-$t$-\ac{DSP} $H$ before composing $H$ with $H_0$.
Standard decomposition algorithms can easily achieve this property; the proof
below also describes an independent linear-time method to transform a non-clean
decomposition tree into a clean one.

\begin{algorithm}[bt]
    \caption{Solve \prob{MCPS1} on \acp{DSP}.}
    \label{alg:opt_cps_sp_graphs}
    \begin{algorithmic}[1]
        \Input{\ac{DSP} $G = (V,E)$, retention ratio $\stretcha \in (0,1)$.}
            \Let{$E'$}{$E$}
            \Let{$T$}{clean series-parallel decomposition tree for $G$}
            \ForAll{tree node $\sigma$ in a bottom-up traversal of $T$}
                \Let{($H$, $(s,t)$)}{graph and terminal pair corresponding to $\sigma$}
                \If{$\sigma$ is a P-composition \textbf{and} $st\!\in\!E(H)$
                \textbf{and} $st$ is covered by $(E'\intersect E(H))\!\setminus\!\{st\}$}
                    \State remove $st$ from $E'$
                \EndIf
            \EndFor
            \Return{$E'$}
    \end{algorithmic}
\end{algorithm}

\begin{theorem}\label{th:alg_cps_dsps}
    \Cref{alg:opt_cps_sp_graphs} solves \prob{MCPS1} on~\acp{DSP} optimally in
    $\bigO(|V|)$~time.
\end{theorem}
\begin{proof}
    We use induction over the clean series-parallel decomposition tree~$T$ of~$G$, maintaining the following invariants: at the end of each
    for-loop iteration with $(H,(s,t))$ as the graph and terminal pair for the
    respective tree node\rnew{~$\sigma$}, $E' \intersect E(H)$ is an optimal \prob{MCPS1} solution for~$H$,
    and all optimal \prob{MCPS1} solutions for $H$ have equal $s$-$t$-capacity.

    Consider a leaf of $T$. The graph $H$ with a single edge $st$ only
    allows one feasible (and hence optimal) solution consisting of its only
    edge.
    The edge is added to $E'$ during the algorithm's initialization and is
    not removed from it before $\sigma$ has been processed.

    Now consider S-compositions and those P-compositions where no edge $st$
    exists in any of the components.
    They produce no additional paths between the endpoints of any edge (which
    are the only vertex pairs that have to be covered, see
    \Cref{th:series_parallel_edge_feasible_cap_spanner}).
    Thus, the feasible (optimal) solutions of $H$ are exactly those that can be
    created by taking the union of one feasible (optimal, respectively) solution
    for each respective component.
    The algorithm proceeds accordingly by keeping $E'$ unchanged.
    Since the components' respective optimal solutions all have the same
    source-sink-capacity (per induction hypothesis), this also holds true for
    their unions, i.e., the optimal solutions of $H$.

    Now consider a P-composition with $st \in E(H)$. As $T$ is clean,
    there are two components~$H_1$ and~$H_2$
    with $E(H_2) = \{st\}$, and $\pathind{G}{s,t} = H$.
    All edges $e \in H_1$ are covered optimally by $E'
    \intersect E(H_1)$ both in $H_1$ and in $H$ since $st \notin
    E(\pathind{H}{e})$.

    \begin{description}
    \item[Case 1:] If one optimal solution for $H_1$ already covers $st$
    in~$H$, then all optimal solutions for $H_1$ do so (as they all have the
    same $s$-$t$-capacity per induction hypothesis).
    Then, the optimal solutions for $H_1$ are exactly the optimal solutions
    for~$H$, and the algorithm finds one of them by keeping its solution
    for~$H_1$ intact and removing~$st$ from~$E'$.
    Note that this removal does not affect the feasibility of $E'$ for subgraphs
    of $G \setminus \pathind{G}{s,t}$ that have already been processed.

    \item[Case 2:] If $st$ is not yet covered by our optimal solution for $H_1$, \rnew{then} it is not
    covered by any optimal solution for $H_1$.
    Our algorithm chooses the edge~$st$ by keeping optimal solutions for
    both~$H_1$ and~$H_2$.
    An optimal solution~$S$ for~$H$ must contain an optimal solution for~$H_1$:
    \rnew{Note that} $S' \be S \setminus \{st\}$ covers all edges of $H_1$.
    If~$S'$ were not optimal, there would exist another solution
    $S''$ that covers all edges and thus vertex pairs of $H_1$
    with $|S''| < |S'|$.
    But $S''' \be S'' \union \{st\}$ is feasible for~$H$ because the capacity
    requirements for vertex pairs in $H$ and $H_1$ only differ by at most one.
    We arrive at $|S'''| = |S''|+1 < |S'|+1 \leq |S|$, a contradiction.

    \rdel{--- }In addition to an optimal solution for $H_1$, we need exactly one more
    edge to increase the $s$-$t$-capacity and cover~$st$ in~$H$:
    this additional edge is either~$st$ itself or another edge from~$H_1$.
    Assume that adding an additional edge $e_1 \in E(H_1)$ (instead of $st$)
    increases the capacity for $st$ or a later source-sink-pair by~1, then~$st$
    by construction does so as well.
    Thus, adding $st$ instead of $e_1$ is never worse; furthermore, all
    optimal solutions for $H$ have the same $s$-$t$-capacity.
    \end{description}

    For the running time, note that a (clean) series-parallel decomposition
    tree~$T$ can be computed and traversed in linear
    time~\cite{DBLP:journals/siamcomp/ValdesTL82}.
    If~$T$ were not clean, it is trivial to establish this property in linear
    time:
    Traverse~$T$ bottom-up; whenever a leaf~$\lambda$ is the child of a P-node, ascend the
    tree as long as the parents are P-nodes. Let~$\varrho$ be the final such
    P-node, and~$\gamma$ the other child of~$\varrho$ that was not part of the
    ascent.
    Swap~$\lambda$ with~$\gamma$.
    Observe that the ascents for different leafs are disjoint, and thus this
    operation requires overall only $\bigO(|T|) = \bigO(|V|)$ time.

    In each step during the traversal in line 3, we can compute the capacity of the
    current source and sink---for both the current solution and $G$ overall---in
    constant time using the values computed in previous steps:
    a single edge is assigned a capacity of 1, and an S-composition
    (P-composition) is assigned the minimum (sum, respectively) of the capacities of
    its two components.
\end{proof}

\Cref{alg:opt_cps_sp_graphs} relies on the fact that all optimal \prob{MCPS1} solutions of a \ac{DSP} have the same sink-source-capacity (we will reuse this property later in \Cref{alg:opt_cps_lsp}, see \Cref{sec:algorithm_lsp}).
However, for non-uniform edge capacities, this property may not hold anymore as soon as a P-composition occurs where one component consists of a single edge---\Cref{fig:cps_algs_counterexample} visualizes this.
Thus, we need to use a different approach to solve the general \prob{MCPS} problem on \acp{DSP}.

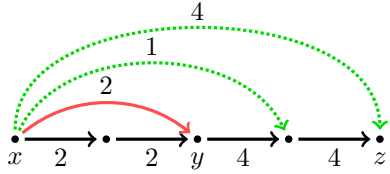
\begin{figure}[t]
\centering
\begin{tikzpicture}[scale=1.20]
        \begin{scope}[every node/.style={dot}]
            \node (a) at (0,0) {};
            \node (b) at (1,0) {};
            \node (c) at (2,0) {};
            \node (d) at (3,0) {};
            \node (e) at (4,0) {};
        \end{scope}
        \node[label, below=0mm of a] (la) {$x$};
        \node[label, below=0mm of c] (lc) {$y$};
        \node[label, below=0mm of e] (le) {$z$};

        \begin{scope}[every edge/.style={edge}]
            \path (a) edge[] node[label,midway,below] {$2$} (b);
            \path (b) edge[] node[label,midway,below] {$2$} (c);
            \path (c) edge[] node[label,midway,below] {$4$} (d);
            \path (d) edge[] node[label,midway,below] {$4$} (e);
            \path (a) edge[draw=myred, bend left=40] node[label,midway,above] {$2$} (c);
            \path (a) edge[draw=mygreen, densely dotted, bend left=65] node[label,midway,above] {$1$} (d);
            \path (a) edge[draw=mygreen, densely dotted, bend left=90] node[label,midway,above] {$4$} (e);
        \end{scope}
\end{tikzpicture}
\caption{A \ac{DSP} with \emph{non-uniform} edge capacities on which the \prob{MCPS1}-focused
    \Cref{alg:opt_cps_sp_graphs} would fail to find an
    optimal \prob{MCPS} solution for $\stretcha = \nicefrac{1}{2}$.
    The \ac{MED} (drawn as straight lines in black) is contained in every
    feasible \prob{MCPS} solution.
    The optimal solution additionally selects the red edge (solid bent), but
    \Cref{alg:opt_cps_sp_graphs} would take the green edges (dotted bent)
    instead.
    This is because it considers the red edge~$xy$ before the greens, and $xy$ is not required to cover itself. %
    In the subsequent composition step, \Cref{alg:opt_cps_sp_graphs} considers the graph without~$z$. This graph has two optimal solutions with different sink-source-capacities: the one with capacity 4 that takes the red edge, and the one with capacity 3 that takes the green edge instead.
    Since the algorithm is now forced to take the solution with lower capacity, it cannot find the optimal solution for the whole graph in the last composition step, but has to select the second green edge as well, to achieve a feasible solution.}
    \label{fig:cps_algs_counterexample}
\end{figure}

\newcommand{\cc}{\ensuremath{\mathfrak{c}}}
\newcommand{\ee}{\ensuremath{\mathfrak{m}}}
\newcommand{\cpair}[1]{\ensuremath{\langle#1\rangle}}
\newcommand{\clist}[1]{\ensuremath{L_{#1}}}
\newcommand{\clistt}[1]{\ensuremath{L'_{#1}}}
\Cref{alg:opt_ccps_sp_graphs} describes a dynamic programming algorithm to find the optimal solution value for general \prob{MCPS} on a \ac{DSP}~$G$ (standard backtracking techniques can then be used to obtain the actual minimum subgraph):
We perform a bottom-up
traversal of the series-parallel decomposition tree of~$G$.
For each processed sub-\ac{DSP}~$G_\sigma$ with source~$s$ and sink~$t$, we
manage a list $\clist{\sigma}$ that stores \emph{description pairs} of the form~\cpair{\ee,\cc}. Each such description pair represents a feasible \prob{MCPS} solution for~$G_\sigma$ with
$\ee$~edges that achieves an~$s$-$t$-capacity of~$\cc$.
Naturally, once the whole graph~$G = G_\varrho$ has been processed, we return the
lowest number of edges of any description pair in~$\clist{\varrho}$ as the minimum objective value.

The algorithm makes use of two subroutines: \textsc{combine} and \textsc{reduce}.
\rnew{First, }\textsc{combine} merges the lists~$\clist{\lambda_1}$
and~$\clist{\lambda_2}$ of description pairs for two decomposition tree nodes~$\lambda_1$
and~$\lambda_2$ to create a preliminary list for their parent node~$\sigma$.
Afterwards, \textsc{reduce} filters out unwanted solutions:
Since we only store description pairs in~$\clist{\sigma}$ that correspond to \emph{feasible}
solutions, we must ensure that the $s$-$t$-capacity of each of these solutions is
higher than $\stretcha\cdot\maxflow{G_\sigma}(s,t)$.
Moreover, to achieve a polynomial (and not just pseudo-polynomial) running time,
we reduce the size of the list by removing every pair that corresponds to
a \emph{capacity-dominated} feasible solution---i.e., a feasible solution such that
a higher $s$-$t$-capacity is achievable with the same number of edges.
This allows us to restrict the size of each list~$\clist{\sigma}$ to $|E(G_\sigma)|$
and thus guarantee a truly polynomial running-time.
One could further speed up the algorithm in practice by also removing
feasible solutions whose $s$-$t$-capacity is achievable with fewer edges,
however, this would not change the theoretical upper bound on the running time.

\begin{algorithm}[t]
    \caption{Solve \prob{MCPS} on \acp{DSP}.}
    \label{alg:opt_ccps_sp_graphs}
    \begin{algorithmic}[1]
        \Input{\ac{DSP} $G = (V,E)$, edge capacities~$\ecap$, retention ratio~$\stretcha$.}
        \Let{$T$}{clean series-parallel decomposition tree for $G$}
        \ForAll{tree node $\sigma$ in a bottom-up traversal of $T$}
            \Let{$(G_\sigma, (s,t))$}{graph and terminal pair corresponding to $\sigma$}
            \If{$|E(G_\sigma)| = 1$}
            \Let{$\clist{\sigma}$}{$\{\cpair{1, \ecap(st)}\}$}
            \Else
                \Let{$(\lambda_1,\lambda_2)$}{children of $\sigma$}
                \If{$\sigma$ is an S-composition}
                    \Let{$\clistt{\sigma}$}{\Call{combine}{$\clist{\lambda_1}$, $\clist{\lambda_2}$, \text{min}}}
                \ElsIf{$\sigma$ is a P-composition with $st\notin E(G_\sigma)$}
                    \Let{$\clistt{\sigma}$}{\Call{combine}{$\clist{\lambda_1}$, $\clist{\lambda_2}$, \text{sum}}}
                \Else \Comment{$\sigma$ is a P-composition with $st \in E(G_\sigma)$, wlog.\ $\{st\}= E(G_{\lambda_2})$}
                    \State{temporarily add \cpair{0,0} to $\clist{\lambda_2}$}
                    \Let{$\clistt{\sigma}$}{\Call{combine}{$\clist{\lambda_1}$, $\clist{\lambda_2}$, \text{sum}}}
                \EndIf
                \Let{$\clist{\sigma}$}{\Call{reduce}{$\clistt{\sigma}$, $\ceil{\stretcha \cdot \maxflow{G_\sigma}(s,t)}$}}
            \EndIf
        \EndFor
        \Let{$\varrho$}{root of $T$}
        \Return{minimum edge cardinality stored in any description pair in $\clist{\varrho}$}\vspace*{3mm}
        \Function{combine}{$\clist{\lambda_1}$, $\clist{\lambda_2}$, \text{combine\_func}}
        \Comment merge two lists
            \Let{$L$}{empty list}
            \ForAll{$\left(\cpair{\ee_1,\cc_1},\cpair{\ee_2,\cc_2}\right)$ in $\clist{\lambda_1} \times \clist{\lambda_2}$}
                \State{add \cpair{\ee_1+\ee_2, \text{combine\_func}(\cc_1, \cc_2)} to $L$}
            \EndFor
            \Return $L$
        \EndFunction\vspace*{3mm}
        \Function{reduce}{$L'$, $\check{c}$}
        \Comment remove infeasible and dominated solutions
            \Let{$L$}{empty list}
            \State{partition pairs in $L'$ into buckets according to their edge cardinality}
            \ForAll{bucket $b$}%
                \Let{$\cpair{\ee,\cc}$}{pair with largest capacity in $b$} %
                \If{$\cc \geq \check{c}$}
                    \State add \cpair{\ee,\cc} to $L$
                \EndIf
            \EndFor
            \Return $L$
        \EndFunction
    \end{algorithmic}
\end{algorithm}

\begin{theorem}\label{th:opt_ccps_sp_graphs}
    \Cref{alg:opt_ccps_sp_graphs} solves \prob{MCPS} on \acp{LSP} optimally in~$\bigO(|V|^3)$ time.
\end{theorem}
\begin{proof}
    We use induction over the clean series-parallel decomposition
    tree~$T$ of~$G$, maintaining the following invariant: at the end of each
    for-loop iteration with $(G_\sigma,(s,t))$ as the graph and terminal pair
    for the respective tree node~$\sigma$, $\clist{\sigma}$ contains exactly the set of pairs corresponding to non-capacity-dominated feasible solutions
    of~$G_\sigma$.
    Observe that multiple solutions may map to the same description pair, but we will see that we do not need to distinguish between these solutions; any solution satisfying a given description pair will suffice.

    To start the induction, let $\sigma$ be a leaf of $T$:
    A graph~$G_\sigma$ with a single edge~$st$ only allows one feasible (and
    hence non-capacity-dominated) solution consisting of its only edge, allowing for
    an~$s$-$t$-capacity of~$\ecap(st)$---the corresponding description pair is added to
    $\clist{\sigma}$.
    Now consider a decomposition tree node~$\sigma$ with terminal pair~$(s,t)$
    and children~$\lambda_1,\lambda_2$, and recall that an
    \prob{MCPS} solution for~$G_\sigma$ only needs to cover all edges in order to
    be feasible (see \Cref{th:series_parallel_edge_feasible_cap_spanner}).
    In contrast, if we were to not cover an edge~$e\in E(G_\sigma)$, even adding all of the edges $E\setminus E(G_\sigma)$ would
    not lead to a solution covering~$e$:
    Composition steps outside of the subtree rooted at~$\sigma$ produce no additional paths between the endpoints of any edge within~$G_\sigma$. Since we assume a \emph{clean} decomposition tree, this is in particular also true for the potential edge $st$, as (in a bottom-up traversal) this edge is the last added edge of $G_\sigma$.

    Assume $st\not\in E(G_\sigma)$.
    The composition step joining $G_{\lambda_1}$ and $G_{\lambda_2}$ will produce no additional paths (and hence not change
    any capacity requirements) between the endpoints of any edge.
    Thus, any feasible (non-capacity-dominated) solution of $\sigma$ can always be constructed as the union of feasible (non-capacity-dominated) solutions for $G_{\lambda_1}$ and~$G_{\lambda_2}$.
    In case of an S-composition (P-composition), the solution's capacity for $\sigma$ is the minimum (sum, respectively) of the capacities of the children's solutions.
    The function \textsc{combine} yields all description pairs corresponding to solutions that can be obtained from
    feasible non-capacity-dominated subsolutions; for the latter, we in turn also only require their corresponding description pairs.
    \textsc{reduce} will shrink that list by removing capacity-dominated solutions. Observe that all solutions will cover the vertex pair $(s,t)$ by construction, and thus the test $\cc\geq \check{c}$ will never fail.

    Finally, assume $st\in E(G_\sigma)$ (and thus $\sigma$ is a P-composition); since we have a clean decomposition tree, we can assume that $E(G_{\lambda_2})=\{st\}$. As before, the edges of $G_{\lambda_1}$ do not gain additional paths in $E(G_\sigma)$, but now $st$ does.
    Any feasible (non-capacity-dominated) solution of $\sigma$ can thus always be constructed as the union of a feasible (non-capacity-dominated) solution for $G_{\lambda_1}$ and a not necessarily feasible solution for $G_{\lambda_2}$: i.e., for $G_{\lambda_2}$ we may consider the solution of selecting $st$, or of not selecting $st$. Our function \textsc{combine} yields all description pairs corresponding to any such solutions, only requiring the description pairs for $\lambda_1$ and the two pairs $\cpair{1,\ecap(st)},\cpair{0,0}$ for $\lambda_2$. Now, \textsc{reduce} will not only remove all capacity-dominated solutions, but also infeasible ones (arising if $st$ was not chosen but necessary to cover $(s,t)$).

    The overall optimal solutions of~$G$ with the highest sink-source-capacity
    are always non-capacity-dominated, and all of them correspond to the same description pair.
    Our invariant guarantees that this pair is stored in~$\clist{\varrho}$
    where $\varrho$ is the root of~$T$.

    Concerning the running time,
    a (clean) series-parallel decomposition tree~$T$ can be computed and
    traversed in linear time~\cite{DBLP:journals/siamcomp/ValdesTL82}.
    After calling \textsc{reduce}, the number of non-capacity-dominated description pairs stored
    in each tree node~$\sigma$ is in~$\bigO(|E(G_\sigma)|) \subseteq \bigO(|E|)$ since
    there is at most one description pair for each edge cardinality in
    $\{1,\dots,|E(G_\sigma)|\}$.
    When combining two of these lists, we obtain~$r\in\bigO(|E|^2)$
    new description pairs, which are then partitioned in $\bigO(r)$ time. The same running time holds for selecting the non-capacity-dominated description pair per bucket, over all buckets.
    Since we have to perform these operations for each of the $\bigO(|V|)$ many decomposition nodes, and $|E| \in \bigO(|V|)$ in \acp{DSP}, we obtain the claimed running time.
\end{proof}

\subsection{MCPS and \texorpdfstring{MCPS$_1$}{MCPS1} on Laminar Series-Parallel Graphs}
\label{sec:algorithm_lsp}

\Cref{alg:opt_ccps_lsp} \rdel{outlines an intuitive}\rnew{implements a natural} approach to solve \prob{MCPS} on an \ac{LSP} $G=(V,E)$. It is based on the
observation that every \ac{LSP} can be partitioned into a set of edge-disjoint
\acp{DSP}:
Consider the \emph{\acp{MEAS}}, i.e., those $\pathind{G}{e}$, for $e \in
E$, such that there is no other edge $e' \in E \setminus \{e\}$ with
$E(\pathind{G}{e})
\subseteq E(\pathind{G}{e'})$.
Since \acp{LSP} are \Cref{th:lsp_psp}-graphs, each of these \ac{MEAS} must be a
\ac{DSP}, and it suffices to cover its edges (see
\Cref{th:series_parallel_edge_feasible_cap_spanner}).
Furthermore, the \ac{EAS}~$\pathind{G}{e''}$ for each edge $e'' \in E$~is
contained in a single \ac{MEAS} (as \acp{LSP} are \Cref{th:lsp_laminar}-graphs).
Hence, one can first identify the \acp{MEAS} and then run \Cref{alg:opt_ccps_sp_graphs} on
each of them individually to obtain an optimal \prob{MCPS} solution.
Note that instead of \Cref{alg:opt_ccps_sp_graphs}, one could employ any blackbox algorithm for \prob{MCPS} on \acp{DSP}.

\newcommand{\MEAS}{\ensuremath{\mathcal{G}_{\text{MEAS}}}}
\begin{algorithm}[t]
    \caption{Solve \prob{MCPS} on \acp{LSP}.}
    \label{alg:opt_ccps_lsp}
    \begin{algorithmic}[1]
        \Input{\ac{LSP} $G = (V,E)$, edge capacities~$\ecap$, retention ratio~$\stretcha \in (0,1)$.}
            \Let{$E'$}{$\emptyset$}
            \ForAll{edge $e = (u,v) \in E(G)$}
                \State identify $\pathind{G}{e}$ via depth-first search starting at $u$ and backtracking from $v$
            \EndFor
            \Let{$\MEAS$}{$\{\pathind{G}{e} \ssep e \in
E \text{ s.t. } \nexists e' \in E \setminus \{e\} \text{ with }
E(\pathind{G}{e})
\subseteq E(\pathind{G}{e'})\}$}
            \ForAll{$\pathind{G}{e} \in \MEAS$}
                \State call \Cref{alg:opt_ccps_sp_graphs} on $(\pathind{G}{e}, \stretcha)$ and add the resulting solution edges to $E'$
            \EndFor
            \Return{$E'$}
    \end{algorithmic}
\end{algorithm}

\begin{corollary}\label{th:opt_ccps_lsp}
    \Cref{alg:opt_ccps_lsp} solves \prob{MCPS} on \acp{LSP} optimally in $\bigO(|E|\cdot|V|^3)$ time.
\end{corollary}
\begin{proof}
    \Cref{alg:opt_ccps_lsp} runs \Cref{alg:opt_ccps_sp_graphs} on the \acp{MEAS} of~$G$ and returns the union of the respective solutions.
    We can consider the different \acp{MEAS} in isolation
    as their solutions are independent of each other:
    \acp{LSP} are \Cref{th:lsp_laminar}-graphs. Thus, for any
    edge $uv\in E(G)$, all $u$-$v$-paths are completely contained in a single \ac{MEAS}.
    Moreover, each \ac{MEAS} must be a \ac{DSP} since \acp{LSP} are \Cref{th:lsp_psp}-graphs.
    As \Cref{alg:opt_ccps_sp_graphs} solves \prob{MCPS} optimally on \acp{DSP} (see \Cref{th:opt_ccps_sp_graphs}), we can conclude that \Cref{alg:opt_ccps_lsp} solves \prob{MCPS} optimally on \acp{LSP} as well.

    For the running time, recall that \acp{LSP} may have a quadratic number of edges, i.e., $|E|\in \bigO(|V|^2)$ and there exist \acp{LSP} with $\Theta(|V|^2)$ edges.
    Consider the three main steps of the algorithm:
    Identifying the \ac{EAS} for a given edge via depth-first search takes $\bigO(|E|)$ time, leading to~$\bigO(|E|^2)$ for the first for-loop. However, observe that each resulting \ac{EAS} is a \ac{DSP} and thus has only $\bigO(|V|)$ edges.
    To identify the \acp{MEAS}, we can perform a $\bigO(|V|)$ inclusion check for every pair of \acp{EAS}, totaling at~$\bigO(|E|^2|V|)$ time.
    Lastly, running \Cref{alg:opt_ccps_sp_graphs} on an \ac{MEAS}~$\pathind{G}{e}$ takes~$\bigO(|V(\pathind{G}{e})|^3)\subseteq \bigO(|V|^3)$ time. Since we have at most $\bigO(|E|)$ \acp{MEAS}, this third step requires $\bigO(|E|\cdot|V|^3)$ time. In particular, we observe that the time required to solve each \ac{MEAS} dominates the running time of identifying them.
\end{proof}

 Within \Cref{alg:opt_ccps_lsp}, we could also use any \prob{MCPS1} algorithm (like \Cref{alg:opt_cps_sp_graphs}) to obtain an optimal \prob{MCPS1} solution. However, for this special case we can give a faster and more straightforward, but functionally equivalent quadratic-time algorithm:

\begin{algorithm}[t]
    \caption{Solve \prob{MCPS1} on \acp{LSP}.}
    \label{alg:opt_cps_lsp}
    \begin{algorithmic}[1]
        \Input{\ac{LSP} $G = (V,E)$, retention ratio $\stretcha \in (0,1)$.}
            \Let{$E'$}{$\emptyset$}
            \ForAll{edge $e \in E$}
            \Let{$\dist{e}$}{number of edges in $\pathind{G}{e}$}
            \EndFor
            \State sort all edges $e \in E$ by non-descending $\dist{e}$
            \ForAll{edge $e = uv \in E$ in order}
                \If{$uv$ is not covered by $E'$}
                    \State add $e$ to $E'$
                \EndIf
            \EndFor
            \Return{$E'$}
    \end{algorithmic}
\end{algorithm}

\begin{theorem}
    \Cref{alg:opt_cps_lsp} solves \prob{MCPS1} on~\acp{LSP} optimally in
    $\bigO(|E|^2)$~time.
\end{theorem}
\begin{proof}
    We prove by induction over $\dist{e}$ that for each
    \ac{DSP} (and hence for each \ac{MEAS}), \Cref{alg:opt_cps_lsp} returns the
    same result as \Cref{alg:opt_cps_sp_graphs}:
    Edges $e$ with $\dist{e} = 1$ (i.e., \ac{MED}-edges)
    are added to~$E'$ by \Cref{alg:opt_cps_lsp} in order to cover themselves.
    Similarly, \Cref{alg:opt_cps_sp_graphs} will add such edges during its
    initialization and never remove them:
    edge $e$ would only be removed if $e$ connected the source and sink of a
    subgraph constructed with a P-composition, a contradiction to $\dist{e} = 1$.

    Now assume that, for some $i \geq 2$, the edges with a $\distsym$-value smaller than $i$ are already processed equivalently to \Cref{alg:opt_cps_sp_graphs}.
    Consider any edge $e = uv$ with $\dist{e} = i$.
    Since $G$ is a \Cref{th:lsp_psp}-graph, $H \be \pathind{G}{e}$ is a
    $u$-$v$-\ac{DSP}.
    As $e \in E(H)$, $H$ can be constructed with a P-composition from two graphs
    $H_1$ and $H_2$ where $E(H_2) = \{e\}$.
    All edges in $H_1$ have already been processed (they have a
    $\distsym$-value smaller than $i$), and the solutions of
    \Cref{alg:opt_cps_lsp} and \Cref{alg:opt_cps_sp_graphs} thus coincide on~$H_1$.
    Hence, both algorithms produce the same solution for~$H$ as they both
    contain~$e$ if and only if~$e$ is not already covered by the current
    solution for~$H_1$.

    For each \ac{MEAS}, \Cref{alg:opt_cps_sp_graphs} and, as we have now shown,
    \Cref{alg:opt_cps_lsp} both find the smallest subgraph that covers all of
    its edges.
    As \acp{LSP} are \Cref{th:lsp_psp}-graphs, this suffices to
    guarantee an optimal \prob{MCPS1} solution for the \ac{MEAS} by
    \Cref{th:series_parallel_edge_feasible_cap_spanner}.
    Moreover, as argued in the proof for \Cref{th:opt_ccps_lsp}, to guarantee an optimal solution for~$G$ as a whole, it suffices to show optimality for the different \acp{MEAS} individually.

    It remains to argue the running time.
    For each $e = uv \in E$, we compute $\dist{e}$ via a depth-first
    search starting at $u$ and counting tree- and cross-edges when backtracking from $v$.
    Overall, this results in $\bigO(|E|^2)$ time.
    The sorting of the edges can be done in $\bigO(|E|)$ time as the domain are
    integer values between $1$ and $|E|$.
    Lastly, to check whether an edge $uv$ is covered, we precompute the
    $s$-$t$-capacity for every edge $st \in E$ on $\pathind{G}{s,t}$
    and then, when needed, compute the $u$-$v$-capacity on the graph $G[E'
    \intersect E(\pathind{G}{u,v})]$ for the current solution~$E'$.
    Note that both of these subgraphs are \acp{DSP} as $G$ is a \Cref{th:lsp_psp}-graph.
    This allows us to compute a series-parallel decomposition tree in $\bigO(|E|)$
    time and traverse it bottom-up to obtain the capacity (cf.\ the proof of
    \Cref{th:alg_cps_dsps}). Doing so twice for every edge takes
    $\bigO(|E|^2)$~time overall.
\end{proof}

\subsection{Applications of \texorpdfstring{\Cref{alg:opt_cps_lsp}}{Algorithm~\ref{alg:opt_cps_lsp}} to Other Problems}
\label{sec:algorithm_other}

Consider the \probl{MSCS} \rdel{(\prob{MSCS})}
problem~\cite{DBLP:journals/siamcomp/KhullerRF95,DBLP:conf/soda/Vetta01}, the
special case of \prob{MED} on strongly connected graphs \rnew{(}\rdel{, i.e., }graphs where
every vertex is reachable from every other vertex\rnew{)}.
Since there are straightforward reductions from \probl{DHC} to \rdel{\prob{MSCS}}\rnew{\probl{MSCS}} to
\prob{MED} to \prob{MCPS1} that all use a common input graph~$G$,
\Cref{alg:opt_cps_lsp} can be adapted to solve these problems as well:
To solve \prob{MED}, one simply has to set~$\stretcha = \min_{(s,t) \in V^2}
\nicefrac{1}{\maxflow{G}(s,t)}$ and then run the algorithm on~$(G,\alpha)$.
Then, \Cref{alg:opt_cps_lsp} does precisely
the same as the algorithm for finding the \ac{MED} \rdel{on}\rnew{in} \acp{DAG}
\cite{DBLP:journals/siamcomp/AhoGU72}: it returns all those edges for which
there is only one path between their endpoints (namely the edge itself).
Hence, our new insight with regards to \rdel{the }\prob{MED} is that the aforementioned
approach does not only solve \prob{MED} optimally on \acp{DAG}, but even on
arbitrary \acp{LSP} as well.
Moreover, if the input graph is strongly connected (Hamiltonian), the returned
\ac{MED} forms \rdel{an \prob{MSCS}}\rnew{a minimum strongly connected subgraph} (directed Hamiltonian cycle, respectively).

\begin{corollary}
    \probl{DHC}, \probl{MSCS} and \probl{MED} can be computed on any \ac{LSP} in quadratic time.
\end{corollary}

\section{Conclusion and Open Questions}
We started research into capacity-preserving subgraphs by
not only showing the NP-hardness of \prob{MCPS} (even with unit edge
capacities on \acp{DAG}), but also presenting a first inapproximability result and proving the para-NP-hardness of \prob{MCPS} w.r.t.\ several (di)graph parameters.
Furthermore, we proposed the new graph class of \acfp{LSP}, a natural extension of
\acfp{DSP} which includes, e.g., cyclic and dense graphs.
We then gave several algorithmic results: two algorithmically surprisingly simple methods to solve \prob{MCPS} with unit edge capacities on \acp{DSP} and \acp{LSP}; a
dynamic programming algorithm for \rdel{the }general \prob{MCPS} on \acp{DSP}; and a method to lift the aforementioned \ac{DSP}-algorithms to \acp{LSP}.

Several questions remain, for example:
Is \prob{MCPS} on undirected graphs (which is a generalization of \probl{MST}) NP-hard?
Is it NP-hard to approximate \prob{MCPS} within a sublogarithmic factor?
Is there a linear-time algorithm for \prob{MCPS} on \acp{LSP}?
All of these questions are still worth exploring even when only
considering unit edge capacities.
Lastly, one may also investigate the hardness of other graph problems on \acp{LSP}.
For example, \rnew{we believe that} our dynamic programming algorithm could \rdel{possibly} be adapted to compute
directed spanners in \acp{LSP}.

\bibliography{main_with_doi}

\end{document}